\documentclass[11pt,letterpaper]{article}

\usepackage{url}
\usepackage{fullpage}
\usepackage{cite}

\usepackage[subtle,mathspacing=normal]{savetrees}
\usepackage{amsmath}
\usepackage{algorithm}
\usepackage[noend]{algpseudocode}
\usepackage{amssymb}
\usepackage{amsthm}
\usepackage{color}
\usepackage{array}
\usepackage{xy}
\usepackage{setspace}
\usepackage{varwidth}
\usepackage{multicol}
\usepackage{algorithm}
\usepackage{hyperref}
\usepackage{thm-restate}
\usepackage{todonotes}
\usepackage{comment}
\usepackage{relate}
\usepackage{enumitem}
\setlist[itemize]{noitemsep}
\setlist[itemize]{nolistsep}
\setlist[enumerate]{noitemsep}
\setlist[enumerate]{nolistsep}

\usepackage{algpseudocode}
\usepackage{mathtools} 
\usepackage{xspace}

\newcommand{\var}{\operatorname{Var}}
\newcommand{\poly}{\operatorname{poly}}

 \newcommand{\E}{\mathbb{E}}

\newcommand{\loglog}[1]{\log\log{#1}}
\newcommand{\polyloglog}{\operatorname{polyloglog}}

\newcommand{\leftchild}[1]{L(#1)}
\newcommand{\rightchild}[1]{R(#1)}

\newcommand{\arr}[1]{\mathcal{A}_{#1}}
\newcommand{\set}[1]{\mathcal{S}_{#1}^0}
\newcommand{\tempset}[1]{\mathcal{S}_{#1}}
\newcommand{\pivot}[1]{\tau_{#1}}
\newcommand{\window}[1]{w_{#1}}
\newcommand{\arrayskew}[1]{Q_{{#1}}}

\newcommand{\insertskew}[1]{\delta_{#1}}

\newcommand{\prob}{\pi}
\newcommand{\otherprob}{\rho}
\newcommand{\insertcount}[1]{\sigma_{#1}}
\newcommand{\insertnumber}[1]{\sigma_{#1}}
\newcommand{\lifetimeinsertcount}[1]{\upsilon_{#1}}
\newcommand{\windowcount}[1]{\mu_{#1}}
\newcommand{\failure}{expensive\xspace}

\newcommand{\inserts}[1]{\mathcal{I}_{{#1}}}

\newcommand{\free}[1]{F_{#1}}

\newcommand{\imax}{k_{\max}}
\newcommand{\calpha}{C_{\alpha}}
\newcommand{\cbeta}{C_{\beta}}

\newcommand{\Deltahat}{\widehat{\Delta}}
\newcommand{\freestart}[1]{F^0_{#1}}
\newcommand{\freediff}[1]{\delta_{#1}}
\newcommand{\pathlength}{d}
\newcommand{\windowpar}[1]{K_{#1}}
\newcommand{\windowtotal}[1]{t_{#1}}
\newcommand{\iskewfinal}[1]{D_{#1}}
\newcommand{\iskewfinaltwo}[2]{D_{#1,#2}}
\newcommand{\arraysize}[1]{m_{#1}}
\newcommand{\numwin}[1]{t}

\newtheorem{theorem}{Theorem}

\newtheorem{lemma}[theorem]{Lemma}
\newtheorem{proposition}[theorem]{Proposition}
\newtheorem{claim}[theorem]{Claim}
\newtheorem{corollary}[theorem]{Corollary}

\theoremstyle{remark}

\newcommand{\defn}[1]{\textbf{\emph{#1}}}
\renewcommand{\paragraph}[1]{\vspace{.2 cm} \noindent \textbf{#1}}

\newcommand{\algo}[1]{\smallskip\noindent\textsc{#1}}

\renewcommand{\epsilon}{\varepsilon}

\title{Nearly Optimal List Labeling}
\date{}

\author{Michael A.~Bender\\\small{Stony Brook University} 
\and Alex Conway\\\small{Cornell Tech} 
\and Mart\'{\i}n Farach-Colton\\\small{New York University}
\and Hanna Koml\'os \\\small{New York University} 
\and Michal Kouck\'y \\\small{Computer Science Institute} \\\small{of Charles University} 
\and William Kuszmaul \\\small{Harvard University} 
\and Michael Saks \\\small{Rutgers University}}

\begin{document}

    \maketitle

\thispagestyle{empty}
\setcounter{page}{0} 

\begin{abstract}
The list-labeling problem captures the basic task of storing a dynamically changing set of up to $n$ elements in sorted order in an array of size $m = (1 + \Theta(1))n$. The goal is to support insertions and deletions while moving around elements within the array as little as possible.

Until recently, the best known upper bound stood at $O(\log^2 n)$ amortized cost. This bound, which was first established in 1981, was finally improved two years ago, when a randomized $O(\log^{3/2} n)$ expected-cost algorithm was discovered. The best randomized lower bound for this problem remains $\Omega(\log n)$, and closing this gap is considered to be a major open problem in data structures. 
    
In this paper, we present the See-Saw Algorithm, a randomized list-labeling solution that achieves a nearly optimal bound of $O(\log n \polyloglog n)$ amortized expected cost. This bound is achieved despite at least three lower bounds showing that this type of result is impossible for large classes of solutions. 
\end{abstract} 

\section{Introduction}\label{sec:intro}

In this paper, we revisit one of the most basic problems in data structures: maintaining a sorted array, as elements are inserted and deleted over time \cite{ItaiKoRo81}. Suppose we are given an array of size $m = (1 + \Theta(1))n$, and a sequence of insertions and deletions, where up to $n$ elements can be present at a time. As the set of elements changes over time, we must keep the elements in sorted order within the array. Sometimes, to support an insertion, we may need to move around elements that are already in the array. The \defn{cost} of an insertion or deletion is the number of elements that we move, and the goal is to achieve as small a cost as possible.\footnote{One might prefer to simply analyze time complexity rather than cost. It turns out that, for the algorithms in this paper, these two metrics will be asymptotically equivalent.}

Since it was introduced in 1981 \cite{ItaiKoRo81}, this problem has been rediscovered in many different contexts \cite{Willard81, Andersson89, GalperinR93, Raman99}, and has gone by many different names (e.g., a sparse-array priority queue \cite{ItaiKoRo81}, the file-maintenance problem \cite{Willard81, Willard82, Willard86, Willard92, BenderFiGi17}, the dynamic sorting problem \cite{kuszmaul2023}, etc). In recent decades, it has become most popularly known as the \defn{list-labeling problem} \cite{dietz1990lower,dietz1994tight, BulanekKoSa12, BenderCFKKW22, Saks18}.

In the decades since it was introduced, the list-labeling problem has amassed a large literature on algorithms \cite{ItaiKoRo81,Willard82, Willard86, Willard92,BenderCoDe02twosimplified, BenderFiGi17,ItaiKa07,GalperinR93, BenderHu07, Katriel02,BenderFiGi05, BenderDeFa05, BenderCFKKW22, BenderBeJo16}, lower bounds \cite{dietz1990lower,dietz1994tight,dietz2004tight,zhang1993density,Saks18,BulanekKoSa12, BulanekKoSa13, BenderCFKKW22}, applications to both theory and practice \cite{Dietz82, BenderCoDe02twosimplified, BenderFiGi17, BenderCoDe02a, BenderDeFa05, BrodalFaJa02, BenderDuIa04, BenderFiGi05, BenderFaKu06, Willard81,Willard86,Willard92,Willard82, WheatmanX21, WheatmanX18, WheatmanB21, PandeyWXB21, LeoB21, LeoB19fastconcurrent}, other parameter regimes for $m$ and $n$ \cite{BabkaBCKS19, AnderssonLa90, zhang1993density, BirdSa07, BulanekKoSa13}, and open problems \cite{Saks18, gal2021computational}. We focus here on some of the major milestones and defer a more in-depth discussion of related work to Section \ref{sec:related}.

\paragraph{Past upper and lower bounds.} The list-labeling problem was introduced in 1981 by Itai, Kohheim, and Rodeh \cite{ItaiKoRo81}, who gave a simple deterministic solution with amortized cost $O(\log^2 n)$. Despite a great deal of interest \cite{ItaiKoRo81,Willard82, Willard86, Willard92,BenderCoDe02twosimplified, BenderFiGi17,ItaiKa07,GalperinR93, BenderHu07, Katriel02,BenderFiGi05, BenderDeFa05}, this bound would remain the state of the art for four decades.

Starting in the early 1990s, much of the theoretical progress was on lower bounds. The first breakthrough came from Dietz and Zhang \cite{dietz1990lower,dietz1994tight}, who showed $\Theta(\log^2 n)$ to be optimal for any \defn{smooth} algorithm, that is, any algorithm that spreads elements out evenly whenever it rebuilds some subarray. Later work by Bul\'anek, Kouck\'y, and Saks \cite{BulanekKoSa12} established an even more compelling claim---that the $\Theta(\log^2 n)$ bound is optimal for any deterministic algorithm. At this point it seemed likely that $\Theta(\log^2 n)$ should be optimal across all algorithms, including randomized ones, but the best lower bound known for randomized solutions, also due to Bul\'anek, Kouck\'y, and Saks \cite{BulanekKoSa13},  remained $\Omega(\log n)$.

Recent work by Bender et al.~\cite{BenderCFKKW22} showed that there is, in fact, a surprising separation between deterministic and randomized solutions. They construct a list-labeling algorithm with $O(\log^{3/2} n)$ expected cost per operation. Their algorithm satisfies a notion of \emph{history independence}, in which the set of array-slot positions occupied at any given moment reveals nothing about the input sequence except for the current number of elements. This history independence property ends up being crucial to the algorithm design and analysis \cite{BenderCFKKW22},\footnote{Roughly speaking, the authors use history independence as a mechanism to avoid the possibility of a clever input sequence somehow ``degrading'' the state of the data structure over time.} but the property also comes with a limitation: Bender et al.~show that any algorithm satisfying this type of history independence \emph{must} incur amortized expected cost $\Omega(\log^{3/2} n)$ \cite{BenderCFKKW22}. 

It remained an open question whether there might exist a list-labeling algorithm that achieves $o(\log^{3/2} n)$ cost, or even $O(\log n)$ cost. Such an algorithm would necessarily need to be non-smooth, randomized, \emph{and} history dependent---and it would need to employ these properties in algorithmically novel ways. 

\paragraph{This paper: nearly optimal list labeling.} In this paper, we present a list-labeling algorithm that achieves amortized expected cost $\tilde{O}(\log n)$ per operation. This matches the known $\Omega(\log n)$ lower bound \cite{BulanekKoSa13} up to a $\poly \log \log n$ factor. We refer to our list-labeling algorithm as the \emph{See-Saw Algorithm}.

Perhaps the most surprising aspect of the See-Saw Algorithm is how it employs history dependence. The algorithm breaks the array into a recursive tree, and attempts (with the help of randomization) to predict which parts of the tree it thinks more insertions will go to. It then gives more slots to the subproblems that it thinks are more likely to get more insertions. 

The idea that such predictions could be helpful would be very natural if we were to assume that our input were either \emph{stochastic} \cite{BenderHu06, BenderHu07} or came with some sort of \emph{prediction oracle} \cite{mccauley2024online}. What is remarkable about the See-Saw Algorithm is that the predictions it makes, and the ways in which it uses them, end up leading to near-optimal behavior on \emph{all possible input sequences}. In fact, to the best of our knowledge, the See-Saw Algorithm is the first example of a dynamic data structure using \emph{adaptivity} to improve the best worst-case (amortized expected) bound on cost for a problem.\footnote{Indeed, one can formalize this claim---it has remained an open question whether there exists \emph{any} data structural problem for which history dependence is necessary to achieve optimal edit-cost bounds \cite{NaorTe01}. Our paper does not quite resolve this problem for the following technical reason: the \emph{lower bound} for history-independent list labeling \cite{BenderCFKKW22} applies only to a weaker notion of history independence than the one in \cite{NaorTe01}.}

Of course, randomization is also important. If the See-Saw Algorithm were deterministic, then the input sequence could easily trick it into making bad decisions.  Thus, it is not just the fact that the algorithm makes predictions based on the past, but also the way in which those predictions interact with the randomness of the algorithm that together make the result possible.\footnote{Interestingly, despite the importance of randomness in our algorithm, the actual \emph{amount} of randomness is relatively small. In fact, one can straightforwardly implement the algorithm using $O((\log \log n)\log n)$ random bits, where $O(\log \log n)$ random bits are used to generate the randomness used within each level of the recursion tree.}

Our result puts the complexity of maintaining a sorted array almost on par with the complexities of other classical sorting problems \cite{meguellati2021survey, alon1996matching, brodal2013survey}. Whether or not the See-Saw Algorithm has a practical real-world counterpart remains to be seen. Such a result could have extensive applications \cite{WheatmanX21, WheatmanX18, WheatmanB21, PandeyWXB21, LeoB21, LeoB19fastconcurrent} to systems that use list labeling as a locality-friendly alternative to binary search trees.

\paragraph{A remark on other parameter regimes. }In addition to the setting where $m = (1 + \Theta(1))n$, list labeling has also been studied in other parameter regimes, both where $m = (1 + \delta)n$ for some $\delta = o(1)$ (our results naturally extend to this regime with cost $\tilde{O}(\delta^{-1} \log n)$), and where $m \gg n$. An interesting feature of the $m \gg n$ regime is that, when $m = n^{1 + \Theta(1)}$, the optimal cost becomes $\Theta(\log n)$, even for randomized solutions \cite{BulanekKoSa13}. Thus, a surprising interpretation of our result is that there is \emph{almost no complexity gap} between the setting where $m = (1 + \Theta(1))n$ and the setting where $m = \poly(n)$. In both cases, the optimal amortized expected complexity is $\tilde{\Theta}(\log n)$. 

\paragraph{Implications to other algorithmic problems. }We remark that there are several algorithmic problems whose best known solutions rely directly on list labeling, and for which list-labeling improvements immediately imply stronger results. 

One significant application is to cache-oblivious B-trees \cite{BenderDeFa05, BrodalFaJa02, BenderDuIa04, BenderFiGi05, BenderFaKu06}, where our list-labeling algorithm can be used to reduce the best known I/O complexity from $O(\log_B N + (\log^{3/2} N)/B)$  \cite{BenderCFKKW22} to $O(\log_B N) + \tilde{O}((\log N)/B)$, which, in turn, reduces to the optimal bound of $O(\log_B N)$ so long as $B = (\log \log n)^{\omega(1)}$.\footnote{Cache-oblivious B-trees make use of so-called packed-memory arrays \cite{BenderDeFa05,BenderDeFa00,BenderFiGi05}, which are list-labeling solutions with the additional property that the array never contains more than $O(1)$ free slots between consecutive elements. As discussed in Section \ref{sec:alg}, our results can be extended to also offer this additional property.}

Another application is to the variation of list labeling in which $n = m$ elements are inserted~\cite{AnderssonLa90,zhang1993density,BirdSa07, BenderCFKKW22} without deletion, that is, an array is filled all the way from empty to full. Here, our results imply an overall amortized expected bound of $\tilde{O}(\log^2 n)$ cost per insertion (see Corollary \ref{cor:fillup}), improving over the previous state-of-the-art of $\tilde{O}(\log^{5/2} n)$ \cite{BenderCFKKW22}.

\paragraph{Paper outline. }The rest of the paper proceeds as follows. We begin in Sections \ref{sec:prelim-salted} and \ref{sec:simplifying} with preliminaries and statements of our main results. We present the See-Saw Algorithm in Section \ref{sec:alg}. We then present the analysis of the algorithm, modulo two central technical claims, in Section \ref{sec:analysis}---these technical claims are proven in Sections \ref{sec:main} and \ref{sec:claim}. Finally, Section \ref{sec:related} gives a detailed review of related work.

\section{Preliminaries}
\label{sec:prelim-salted}

\paragraph{Defining the list-labeling problem. }In the list-labeling problem, there are two parameters, the array size $m$, and the maximum number of elements $n$. We will be most interested in the setting where $m = (1 + \Theta(1))n$, but to be fully general, we will also allow $m = (1 + \delta) n$ for $\delta = o(1)$.

We are given an (online) sequence of insertions and deletions, where at any given moment there are up to $n$ elements present. The elements are assumed to have a total order, and our job is to keep the current set of elements in sorted order within the size-$m$ array. As insertions and deletions occur, we may choose (or need) to move elements around within the array. The \defn{cost} of an insertion or deletion is defined to be the number of elements that get moved.

When discussing randomized solutions, one assumes that the input sequence is generated by an oblivious adversary. In other words, the input sequence is independent of the random bits used by the list-labeling algorithm.

\paragraph{Conventions.} To simplify discussion throughout the paper, we will generally ignore rounding issues. Quantities that are fractional but should be integral can be rounded to the closest integer without affecting the overall analysis by more than a negligible error.

We will always be interested in bounding \emph{amortized expected cost}. A bound of $C$ on this quantity means that, for all $i$, the expected total cost of the first $i$ operations is $O(i C)$.

\section{Main Results}\label{sec:simplifying}

Formally, the main result of this paper is that:

\begin{restatable}{theorem}{actualmain}
For $\delta \in (0, 1)$, and $m = (1 + \delta)n$, there is a solution to the list-labeling problem on an array of size $m$, and with up to $n$ elements present at a time, that supports amortized expected cost $O(\delta^{-1} (\log n) (\log \log n)^3)$ per insertion and deletion.
\label{thm:actualmain}
\end{restatable}

\begin{corollary}
    If $m = (1 + \Theta(1))n$, then there is a solution to the list-labeling problem with amortized expected cost $O((\log n) (\log \log n)^3)$ per operation.
\end{corollary}

In Appendix \ref{app:reductions}, we present a series of standard w.l.o.g.~reductions that together reduce the task of proving Theorem \ref{thm:actualmain} to the task of proving the following equivalent but simpler-to-discuss result:

\begin{restatable}{theorem}{thmnice}
Let $m$ be sufficiently large, and $n = m / 2$. The See-Saw Algorithm is a list-labeling algorithm that, starting with $m / 4$ elements, can support $m/4$ insertions with amortized expected cost $O((\log n) (\log \log n)^3)$.
    \label{thm:nicetoprove}
\end{restatable}

Note that, compared to Theorem \ref{thm:actualmain}, Theorem \ref{thm:nicetoprove} is able to assume an insertion-only workload, a relationship of $m = 2n$, and a starting-state of $m / 4 = n / 2$ elements. The rest of the paper will be spent proving Theorem \ref{thm:nicetoprove}.

Finally, we remark that, in Appendix \ref{app:reductions}, we also arrive at the following corollary:
\begin{restatable}{corollary}{fillup}
    There is a list-labeling algorithm that inserts $n$ items into an initially empty array of size $n$ with amortized expected cost $O((\log^2 n) (\log \log n)^3)$. 
    \label{cor:fillup}
\end{restatable}

\section{The See-Saw Algorithm}\label{sec:alg}

In this section, we present the \defn{See-Saw Algorithm}, which we subsequently prove achieves $\tilde{O}(\log n)$ amortized expected cost per insertion on an array $\arr{}$ of size $m$. We also present detailed pseudocode for the algorithm in Appendix \ref{app:pseudocode}. 

As discussed in Section \ref{sec:simplifying}, we will consider, without loss of generality, that we have an insertion-only workload, that the initial number of elements is $m/4$, and that we are handling $m/4$ total insertions. The algorithm will make use of parameters $\alpha = \calpha (\log \log n)^2$ and $\beta = \cbeta (\log \log n)^2$, where $\calpha$ and $\cbeta$ are positive constants selected so that $\calpha$, $\cbeta$, and $\calpha / \cbeta$ are all sufficiently large.

\paragraph{Defining a subproblem tree.} At any given moment, we will break the array into a recursive \defn{subproblem tree}. Each subproblem $\prob$ in the tree is associated with a subarray $\arr{\prob}$ whose size is denoted by $m_\prob = |\arr{\prob}|$. For the subproblem $\prob$ at the root of the tree, $\arr{\prob}$ is the entire array. Each non-leaf node $\prob$ has left and right children, $\leftchild{\prob}$ and $\rightchild{\prob}$ respectively, such that $\arr{\prob} = \arr{\leftchild{\prob}} \oplus \arr{\rightchild{\prob}}$ (the concatenation of 
$\arr{\leftchild{\prob}}$ and $\arr{\rightchild{\prob}}$). In contrast with the classical $O(\log^2 n)$ algorithm (and, indeed, all previous algorithms that we are aware of), the structure of the subproblem tree used by the See-Saw Algorithm will be \emph{non-uniform}, meaning that sibling subproblems $\leftchild{\prob}$ and $\rightchild{\prob}$ will not necessarily satisfy $\arraysize{\leftchild{\prob}} = \arraysize{\rightchild{\prob}}$.

As we shall see, the structure of the tree will evolve over time, with subproblems getting terminated and then replaced by new ones. When a subproblem $\prob$ is first created, we will use $\set{\prob}$ to refer to the set of elements stored in $\arr{\prob}$ when $\prob$ is created.

Because subproblems are created and destroyed over time, the children $L(\prob)$ and $R(\prob)$ of a given subproblem $\prob$ may get replaced many times during $\prob$'s own lifetime. Thus one should think of $L(\pi)$ and $R(\pi)$ as time-dependent variables, referring to $\pi$'s \emph{current} left and right children at any given moment.

\paragraph{How an insertion decides its root-to-leaf path.} Given an insertion $x$ that goes to a subproblem $\prob$, the protocol for determining which child $\leftchild{\prob}$ or $\rightchild{\prob}$ the insertion $x$ goes to can be described as follows: 
if $\leftchild{\prob}$ contains at least one element, and if $\max_{y \in L(\prob)} y > x$,
then $x$ is placed in $\leftchild{\prob}$; otherwise, it goes to $\rightchild{\prob}$.
\footnote{Assuming we start with $m/4$ elements in the array, it will turn out that the left children are never empty---we will not need to formally prove this fact, however, as it does not end up being necessary for our analysis.} This rule determines the root-to-leaf path that a given insertion takes. 

\paragraph{Implementing leaves. } When a subproblem is created, there are two conditions under which it is declared to be a leaf: subproblems $\prob$ whose initial density $|\set{\prob}| / \arraysize{\prob}$ is greater than $3/4$ are \defn{\failure leaves}; and subproblems $\prob$ whose subarray satisfies $\arraysize{\prob} \le 2^{\sqrt{\log n}}$ are \defn{tiny leaves}.

In both cases, leaf subproblems $\prob$ are implemented using the classical algorithm of Itai, Konheim and Rodeh\cite{ItaiKoRo81}, whose cost per operation is $O(\log^2 \arraysize{\prob})$. 
For tiny leaves, this results in $O(\log n)$ amortized cost per operation. For \failure leaves, this could result in as much as $O(\log^2 m) = O(\log^2 n)$ cost per operation. One of the major tasks in analyzing the algorithm will be to bound the total cost incurred in \failure leaves over all operations.

\paragraph{Initializing a subtree.}
When a subproblem $\prob$ is first initialized, it is always initialized to be \defn{balanced}. This means: (1) that the elements in $\prob$ are evenly distributed across $\arr{\prob}$; and (2) that, within each level of the subtree rooted at $\prob$, all of the subproblems within that level have arrays that are the same sizes as each other. 

Thus we define the \texttt{CreateSubtree}$(\mathcal{A}', \mathcal{S}')$ procedure as follows. The procedure takes as input a subarray $\mathcal{A}'$ and a set $\mathcal{S}'$ of elements, and it produces a tree of balanced subproblems, where the root of the tree $\prob$ satisfies $\arr{\prob} = \arr{}'$ and $\set{\prob} = \mathcal{S}'$. 
This can be accomplished by first spreading the elements $\mathcal{S}'$ evenly across the array $\arr{}$, and then creating a subproblem $\prob$ satisfying $\arr{\prob} = \mathcal{A}$. 
If $\prob$ is a leaf, then this is the entire procedure. Otherwise, if $\prob$ is not a leaf, then we create children for $\prob$ with sub-arrays of size $m_\prob / 2$; and if those are not leaves, we create grandchildren for $\prob$ with subarrays of size $m_{\prob} / 4$; and so on.

\paragraph{Implementing non-leaf subproblems. }
Now consider a non-leaf subproblem $\prob$ and let $\inserts{\prob}$ denote the sequence of insertions that $\prob$ receives. 

The first thing that $\prob$ does is select a \defn{rebuild window size} $w_\prob$---this window size is selected from a carefully constructed probability distribution that we will describe later on. The subproblem $\prob$ then treats the insertion sequence $\inserts{\prob}$ as being broken into equal-sized \defn{rebuild windows} $\inserts{\prob, 1}, \inserts{\prob, 2}, \ldots$, where each rebuild window $\inserts{\prob, i}$ consists of up to $w_\pi$ insertions. (Only the final rebuild window may be smaller).

Whenever one rebuild window $\inserts{\prob, i}$ ends and another $\inserts{\prob, i + 1}$ begins, $\prob$ performs a \defn{rebuild}. The rebuild terminates all of $\pi$'s descendant subproblems, spends $O(m_\prob)$ cost on rearranging the elements within $\arr{\pi}$, and then creates new descendant subproblems for $\prob$. 

To describe this rebuild process, let us refer to $\prob$'s children before the rebuild as $\overline{L}(\prob), \overline{R}(\prob)$ and to $\prob$'s new children after the rebuild as $L(\prob), R(\prob)$. 
The most interesting step in the rebuild is to select the sizes $m_{L(\prob)}$ and $m_{R(\prob)}$ for $\mathcal{A}_{L(\prob)}$ and $\mathcal{A}_{R(\prob)}$---we will describe this step later in the section. After selecting the size $m_{L(\prob)}$, the new subproblem $L(\prob)$, along with its descendants, are created by calling \texttt{CreateSubtree}$(\arr{L(\prob)}, \set{L(\prob)})$, where $\set{L(\prob)}$ is the same set of elements that were stored in $\overline{L}(\prob)$ prior to the rebuild. Similarly, $R(\prob)$ and its descendants are created by calling \texttt{CreateSubtree}$(\arr{R(\prob)}, \set{R(\prob)})$, where $\set{R(\prob)}$ is the set of elements that was stored in $\overline{R}(\prob)$ prior to the rebuild. 

It is worth emphasizing that the rebuild changes the \emph{sizes} of the subarrays used to implement each of $\pi$'s children, but does not change the \emph{sets} of elements stored within the two children. Furthermore, although $\prob$'s new child subtrees are initialized to be balanced, the subtree rooted at $\prob$ need not be balanced: $\arraysize{\leftchild{\prob}}$ need not be equal to $\arraysize{\rightchild{\prob}}$, nor does $|\set{\leftchild{\prob}}|$ need to equal $|\set{\rightchild{\prob}}|$.  The cost of such a rebuild is $O(\arraysize{\prob})$.

It remains to specify how to choose $\window{\prob}$ (the rebuild-window size), and how to choose $\arraysize{\leftchild{\prob}}$ and $\arraysize{\rightchild{\prob}}$ during each rebuild.  For this second point, rather than setting $\arraysize{\leftchild{\prob}} = \arraysize{\rightchild{\prob}} = \arraysize{\prob}/2$, $\prob$ will (sometimes) try to \emph{predict} which of $\leftchild{\prob}$ or $\rightchild{\prob}$ will need more free slots in the future, and will potentially give a different number of slots to each of them. Successfully predicting which subproblem will get more insertions is key to our algorithm and is described more fully below.

\paragraph{A key component: selecting window sizes.} We now describe how $\prob$ selects its rebuild-window size $\window{\prob}$. This turns out to be the \emph{only} place in the algorithm where randomization is used. We start by generating a random variable $\windowpar{\prob}$, taking values in $[0,\imax]$ where $\imax=2\loglog{n}$ and where
$p_k=\Pr[\windowpar{\prob}=k]$ is given by:
\begin{eqnarray*}
p_k & = & 2^{-(k+1)}\left(1+\frac{k}{\imax}\right)\hspace{.5 in}\text{for $k\in [1,\imax]$}\\
p_0 & = & 1-\sum_{k=1}^{\imax} p_k \le 1/2.
\end{eqnarray*}
Having drawn $\windowpar{\prob}$ from this distribution, we then set
$$\window{\prob} = \frac{\arraysize{\prob}}{\alpha 2^{\windowpar{\prob}}},$$
where $\alpha =\Theta((\loglog{n})^2)$ is the parameter defined at the beginning of the section.  As we shall see, the specific details of this probability distribution will end up being central to the analysis of the algorithm.

\paragraph{A key component: selecting array skews.}
Next we describe how $\prob$ selects the sizes of the subarrays $\arraysize{\leftchild{\prob}}$ and $\arraysize{\rightchild{\prob}}$ to be used by its children $\leftchild{\prob}$ and $\rightchild{\prob}$ within a given rebuild window. Here, critically, $\prob$ \emph{adapts} to the history of how insertions in past rebuild windows behaved. 

The rebuilds of $\prob$ will behave differently at the beginning of odd-numbered windows than even-numbered ones. At the beginning of odd-numbered rebuild windows, $\prob$ will not make any attempt to adapt; it will simply set $\arraysize{\leftchild{\prob}} = \arraysize{\rightchild{\prob}} = \arraysize{\prob}/2$. (This, incidentally, is why $\prob$ need not do any rebuilding at the beginning of the first rebuild window). To adapt at the beginning of an even-numbered rebuild window $j$, $\prob$ will count the total number of insertions that went right minus the total number that went left during rebuild window $(j-1)$---call this quantity the \defn{insertion skew}  $\iskewfinaltwo{\prob}{j-1}$. Then, at the beginning of rebuild window $j$, $\prob$ will set rebuild-window $j$'s \defn{array skew}
$$\arrayskew{\prob, j} = \arraysize{\prob} \cdot \frac{\iskewfinaltwo{\prob}{j - 1}}{\beta \window{\prob}},$$
where $\beta = \Theta((\log \log n)^2)$ is the parameter defined at the beginning of the section. Finally, using this array skew, $\prob$ sets $\arraysize{\leftchild{\prob}} = \arraysize{\prob}/2 - \arrayskew{\prob, j}$ and $\arraysize{\rightchild{\prob}} = \arraysize{\prob}/2 + \arrayskew{\prob, j}$. 

In other words, $\prob$ uses odd-numbered rebuild windows for \emph{learning}, and then uses even-numbered rebuild windows for \emph{making use of what it has learned}. Within the even-numbered rebuild windows, $\prob$ is essentially trying to predict which subproblem will get more insertions, and then giving that subproblem more slots. 

In the same way that the window-size selection is the only place the in the algorithm that makes use of randomization, the selection of array skews is the only place that \emph{adapts} to the historical behavior of the input. Of course, if the See-Saw Algorithm were deterministic, it would be easy to construct an insertion sequence that would defeat this type of adaptivity. Thus, it is not just the fact that the See-Saw Algorithm tries to make predictions based on the past, but also the way in which this interacts with the (randomly chosen) window-size $\window{\prob}$, that together make the algorithm work.
The analysis of how adaptivity and randomization work together to minimize \failure leaves will require a number of deep technical ideas, and will constitute the main technical contribution of the paper. 

\paragraph{One final source of cost: subproblem resets. } So far, we have encountered one way in which a subproblem $\prob$'s life can end, namely, that one of its ancestors begins a new rebuild window. There will also be another way in which $\prob$'s life can end: If $\prob$ receives a total of $\arraysize{\prob} / \alpha = \Theta(\arraysize{\prob} / (\loglog{n})^2)$ insertions, then $\prob$ will be \defn{reset}.\footnote{We remark that, in our pseudocode in Appendix~\ref{app:pseudocode}, the \emph{parent} of $\prob$ is responsible for implementing resets (with the exception of the case where $\prob$ is the root, which is handled separately).} This threshold $\arraysize{\prob} / \alpha$ for the maximum number of insertions that $\prob$ can handle before being reset is referred to as its \defn{quota}.
What it means for a subproblem $\prob$ to be reset is that it (and its descendants) are terminated, and that a new balanced subproblem tree is created in $\prob$'s place using the \emph{same} subarray and the \emph{same} set of elements as $\prob$ did. The new subtree is created using the \texttt{CreateSubtree} protocol. 

One should think of resets as, in some sense, being a technical detail. They are just there to ensure that each subproblem has a bounded number of insertions. The real engine of the algorithm, however, is in the implementation of rebuilds.

\paragraph{A remark on non-smoothness, randomness, and history dependence.} As discussed in the introduction, there are three properties that past work has already shown to be necessary if one is to achieve $o(\log^{1.5} n)$ overall amortized expected cost.  
These properties are \emph{non-smoothness}~\cite{dietz1990lower, dietz2004tight}, \emph{randomness}~\cite{BulanekKoSa12}, and \emph{history dependence}~\cite{BenderCFKKW22}. It is therefore worth remarking on their roles in the See-Saw Algorithm.

The randomness in the algorithm is used to select the rebuild window size $\window{\prob}$ for each subproblem. We will see that, although the input sequence can \emph{attack} the algorithm for one specific choice of $\window{\prob}$, there is no way for it to systematically attack $\window{\prob}$ across the entire distribution from which it is selected.  

The fact that our algorithm is non-history-independent, and the fact that the rebuilds it performs are non-smooth, are both due to the same step in the algorithm: the step where, at the beginning of each even-numbered rebuild window $j$, $\prob$ selects the array skew $\arrayskew{\prob,j}$ adaptively based on what occurred during the previous rebuild window. This adaptivity is fundamentally history dependent, and then the rebuild that it performs on $\arr{\prob}$ is fundamentally non-smooth (since, for a given rebuild, there is only one possible value for the array skew that would result in the rebuild being smooth). 

\paragraph{A remark on how to think about the range of values for array skews.} It is worth taking a moment to understand intuitively the range of possible values for the array skew $\arrayskew{\prob, j}$. Since $|\iskewfinaltwo{\prob}{j - 1}| \le \window{\prob}$, the array skew will always satisfy $|\arrayskew{\prob, j}| \le \arraysize{\prob} / \beta = O(\arraysize{\prob} / (\log \log n)^2)$. So, perhaps surprisingly, there is a sense in which the array skew is always a \emph{low-order term} compared to the size of the array $\arraysize{\prob}$. On the other hand, the window size $\window{\prob}$ is also at most $\arraysize{\prob} / \alpha = O(\arraysize{\prob} / (\log \log n)^2)$, so one should think of the maximum possible window size $\window{\prob}$ as being comparable to the maximum possible array skew $\arrayskew{\prob, j}$ (and, in fact, the former quantity is the smaller because $\alpha > \beta$).

\paragraph{A remark on packed-memory arrays.} Many data-structural applications of list-labeling require the additional property that there are at most $O(1)$ free slots between any two consecutive elements in the array. A list-labeling solution with this property is typically referred to as a \defn{packed-memory array} \cite{BenderDeFa05,BenderDeFa00,BenderFiGi05}. We remark that the See-Saw Algorithm can be turned into a packed-memory array with the following modification: Whenever the initial density of a non-leaf subproblem $\prob$ is less than, say, $0.25$, we automatically set all of the array skews $\arrayskew{\prob, 1}, \arrayskew{\prob, 2}, \ldots$ to $0$. This turns out to not interfere with the analysis of the See-Saw Algorithm in any way, since as we shall see, the analysis only cares about the array skews in cases where the initial density is at least $0.5$ (Lemma \ref{lemma:main}). On the other hand, with this modification in place, no subproblem $\prob$ is ever given fewer than $(0.25 - o(1)) m_{\prob}$ elements, which implies that we have a packed-memory array.

\section{Algorithm Analysis}
\label{sec:analysis}

In this section, we prove Theorem \ref{thm:nicetoprove}, which, as discussed in Section \ref{sec:simplifying}, implies the main result of the paper, Theorem \ref{thm:actualmain}. We begin by restating Theorem \ref{thm:nicetoprove} below.

\thmnice*

The proof of Theorem \ref{thm:nicetoprove} occupies
this section and the next two. In this section we prove the
theorem assuming two results,  Lemma~\ref{lemma:main}
and Claim~\ref{claim:conc}. These are proved in the following
two sections.

In the algorithm description, at any
point in time there is a binary tree of subproblems. It is
important to keep in mind that 
the tree is dynamic; subproblems are terminated and new ones
are created in their place. As a convention, we will refer to the subproblems that, over time, serve as the roots of the tree as the \defn{global subproblems}.

\vspace{.2 cm}

\noindent Throughout the section, we will make use of the following notation for discussing a subproblem $\prob$, some of which were also defined in Section \ref{sec:alg}:
\begin{itemize}
\item $\arraysize{\prob} = |\arr{\prob}|$.
\item $s_{\prob} = |\set{\prob}|$, where $\set{\prob}$ is the set of elements in $\prob$ at the \emph{beginning} of its lifetime.
\item $\inserts{\prob}$ is the full sequence of inserts that arrive to subproblem $\prob$ during its lifetime.
\item $\inserts{\prob,j} \subseteq \inserts{\prob}$ is the subsequence of inserts that arrive to $\prob$ during its $j$-th rebuild window.
\item $\arrayskew{\prob,j}$ is the value of the array skew used for $\prob$'s $j$-th rebuild window.
\item $\iskewfinal{\prob}(v)$ where $v \in \inserts{\prob}$ is equal to 1 if $\prob$
sends $v$ right and $-1$ if $\prob$ sends $v$ left.
\item $\iskewfinal{\prob}(J)$, where $J \subseteq \inserts{\prob}$ is equal
to $\sum_{v \in J} \iskewfinal{\prob}(v)$, which is the number
of elements of $J$ that $\prob$ sent right minus the number that were sent left.
\item $\iskewfinal{\prob} = \iskewfinal{\prob}(\inserts{\prob})$
is the total number of inserts to $\prob$ that went right minus the number that went left.
\item $\iskewfinal{\prob,j}=\iskewfinal{\prob}(\inserts{\prob,j})$.
\item $\insertnumber{\prob} = |\inserts{\prob}|$ is the total number of elements that are inserted into $\prob$ during its lifetime (not including the $s_{\prob}$ elements initially present).  Note
that, by design, $\insertnumber{\prob}\leq \arraysize{\prob}/\alpha$. 
\item $\free{\prob} = 1-\frac{\insertnumber{\prob}+s_{\prob}} {\arraysize{\prob}}$ is the density of free slots in $\arr{\prob}$ at the \emph{end} of $\prob$.
\item $\freestart{\prob}=1- \frac{s_{\prob}}{\arraysize{\prob}}$ is the density of free slots in $\arr{\prob}$ at the \emph{beginning} of $\prob$.
\item $\windowpar{\prob}$ is the value of the random integer that determines the rebuild window size $\window{\prob} = m_{\prob}/(\alpha 2^{\windowpar{\prob}})$. 
\item $\windowtotal{\prob}$ is the total number of rebuild windows that $\prob$ starts over its
lifetime.
\end{itemize}

\vspace{.2 cm}

We will often drop the subscript $\prob$, e.g., on $\arr{\prob}, \set{\prob}, \inserts{\prob,j}$ and $\arrayskew{\prob,j}$,
when the subproblem is clear from context. 
However, when we write $m$, we always mean the full array size.

We organize the set of all subproblems
that exist throughout the algorithm into a nonbinary tree called the \defn{history tree}.
For a given subproblem $\prob$, its children will be
all left and right subproblems that it ever creates.  (The root
of the tree is a fictitious root subproblem, and the children of
the root are the global subproblems.)
A subproblem $\prob$ will have at least $\windowtotal{\prob}$ 
different left subproblems and $\windowtotal{\prob}$ right subproblems, since it starts a new left and right subproblem
at the beginning of each rebuild window.  (Recall that $\windowtotal{\prob}$ is the total number of rebuild windows of $\prob$.) A subproblem may have more than
one left or right child subproblem per rebuild window because a child subproblem may reach
its quota, which causes it to reset, causing it to get replaced by a new subproblem.
The leaves of the history tree are the {\failure}-leaf and tiny-leaf subproblems.

Note that, for a given subproblem $\prob$ in the history tree,
the number of children $\prob$ has is not fixed in advance
but depends both on the random choices of window sizes by $\prob$ and its ancestors,
and also on the specific insertion sequence (the set of which subproblems get terminated because they reach their quotas may depend on the specific insertion sequence).

\subsection{The Basics: Proving Correctness, and Bounding the Costs of Rebuilds, Resets, and Tiny Leaves} 

We start with some basic observations:

\begin{proposition}
\label{prop:basic}
For any subproblem $\prob$:
\begin{enumerate}
\item If $\prob$ is non-global, then $\arraysize{\prob} \in [0.49 \arraysize{\otherprob},0.51\arraysize{\otherprob}]$, where $\otherprob$ is the parent of $\prob$.
\item The total number of items $s_{\prob}+\insertnumber{\prob}$ that $\prob$ must store in its subarray
is at most $0.8\arraysize{\prob}$.
\end{enumerate}
\end{proposition}
\begin{proof}
For the first part, if $\prob$ is inside the $j$-th
rebuild window of $\otherprob$ then the size of $\prob$'s array is
$\frac{1}{2}\arraysize{\otherprob}\pm |\arrayskew{\otherprob,j}|$
and  $|\arrayskew{\otherprob,j}| \leq \frac{|\iskewfinal{\otherprob,j-1}|\arraysize{\otherprob}}{\window{\otherprob}\beta} \leq {\arraysize{\otherprob}/\beta} \leq 0.01\arraysize{\otherprob}$
 (since
$\beta \geq 100$).

For the second part, the assertion is true for any global subproblem since the total number of elements in the array never exceeds $\arraysize{\prob}/2$.
For a non-global subproblem $\prob$ with parent $\otherprob$, it must be that $\otherprob$ is not an \failure leaf (since \failure leaves don't initiate subproblems), so $s_{\otherprob} \leq 0.75 \arraysize{\otherprob}$. Assume without loss of generality that $\prob$ is a left subproblem of $\otherprob$. Recall that, when $\otherprob$ is created, it gives half of the elements in $\set{\otherprob}$ to its left child, and that whenever $\otherprob$ rebuilds its children, it does not move any elements between them (it just changes the sizes of their arrays); thus, the number of elements from $\set{\otherprob}$ that $\prob$ contains is just $|\set{\otherprob}|/2 = s_{\otherprob}/2$. So the number of elements $s_\prob + \insertnumber{\prob}$ in $\prob$ at the end of $\prob$'s lifetime is at most
$$s_\otherprob / 2 + \insertnumber{\otherprob} \le 0.38 \arraysize{\otherprob} + \arraysize{\otherprob} / \alpha \le 0.39 \arraysize{\otherprob}.$$
By the first part of the proposition, we have $\arraysize{\otherprob} \le \arraysize{\prob} / 0.49$, so our bound on $s_\prob + \insertnumber{\prob}$ is at most
$$0.39 (\arraysize{\prob} / 0.49) \le 0.8 \arraysize{\prob}.$$
\end{proof}

As a corollary, we can establish the correctness of the See-Saw Algorithm.

\begin{corollary}
    The See-Saw Algorithm is a valid list-labeling algorithm.
\end{corollary}

\begin{proof} 
In the algorithm, each successive insert is passed down the current subproblem tree to a leaf subproblem
which inserts the item into its subarray using the classical algorithm.  The classical algorithm at
leaf subproblem $\prob$ will fail
to carry out an insertion only if the total number of items assigned to $\prob$ exceeds
$\arraysize{\prob}$, but the final part of Proposition ~\ref{prop:basic} ensures that this
does not happen. The only other times that items are moved in the array are when a non-leaf problem
does a rebuild of one or both of its subproblems.  Such a rebuild will fail only
if for a created subproblem $\otherprob$ the number of items $s_{\prob}$ initially 
assigned to $\otherprob$ exceeds $\arraysize{\prob}$, which again is impossible by the last
part of Proposition ~\ref{prop:basic}.

The above guarantees that after each insertion, all  items inserted so far are placed in the array.
It remains to verify that the ordering of the items in the array is consistent with the intrinsic
ordering on items. At any point in the execution, if $\prob$ is an active leaf subproblem then
the items in the subarray of $\prob$ are in order by the correctness of the classical algorithm.
If two items are assigned to different leaves $\prob$ and $\otherprob$
with $\prob$ to the left of $\otherprob$ (under the uusal left-to-right ordering of leaves)
then it is easy to see from the definition of the algorithm that  the subarray of $\prob$ is entirely
to the left of the subarray of $\otherprob$ and the items assigned to $\prob$ are all less than
the items assigned to $\otherprob$, so the two items will be in correct order.
\end{proof}

It remains to bound the cost of the algorithm.
Define the \defn{level} of a subproblem to be its depth in the history tree, where global subproblems are said to have level $1$. The first part of Proposition~\ref{prop:basic} implies that the maximum level of any subproblem is at most $1.5\log m \le 2 \log n$.

The cost incurred by the data structure can be broken into four groups: (1) The cost of rebuilds, which occur every time that a non-leaf subproblem $\prob$ finishes a rebuild window and begins a new one; (2) the cost of resets, which which occur whenever a subproblem reaches its quota for the total number of insertions it can process; (3) the cost of tiny subproblems; and (4) the cost of \failure leaf subproblems.

We can bound the first three of these with the following lemma:
\begin{lemma}
\label{lemma:cost}
     The total expected amortized cost (across all subproblems) from rebuilds, resets, and tiny subproblems is $O((\log n) (\log\log n)^3)$ per insertion.
    \label{lem:simplecosts}
\end{lemma}

\begin{proof}
First we bound the cost of all resets.  A reset is done
when a subproblem  $\prob$ has reached its
quota of $\arraysize{\prob}/\alpha$ insertions, and the cost of the reset is $\arraysize{\prob}$. We can bound the sum of these reset costs by charging $\alpha$ to each insertion that went through $\prob$. Overall, each insertion travels through $O(\log n)$ total subproblems, and therefore gets charged $O(\alpha \log n) = O((\log n) (\log \log n)^2)$. Thus the amortized expected cost of resets is $O((\log n) (\log \log n)^2)$.

Next we bound the cost of rebuilds. The number of rebuilds that a subproblem $\prob$ performs is $\windowtotal{\prob} -1 = \lfloor (\insertnumber{\prob} - 1) / \window{\prob} \rfloor$ where $\windowtotal{\prob}$ is its number of rebuild windows.  (Remember that, crucially, a subproblem does not perform a rebuild at the beginning of its first window, as the elements in $\arr{\prob}$ are already evenly spread out at that point in time, which is the state that $\prob$ initially wants.) Recall that
$\window{\prob}=\arraysize{\prob}/(\alpha 2^{\windowpar{}})$ where $\windowpar{}=\windowpar{\prob}$.
If $\windowpar{}=0$ then $\window{\prob}=\arraysize{\prob}/\alpha$ which is precisely $\prob$'s quota, so the number of rebuilds is $\lfloor (\insertnumber{\prob} - 1) / \window{\prob} \rfloor \le \lfloor (\arraysize{\prob}/\alpha - 1) / \window{\prob}\rfloor = 0$.
For $\windowpar{}\geq 1$ the number of rebuilds performed by $\prob$ is at most $\insertnumber{\prob}/(\arraysize{\prob}/(\alpha 2^{\windowpar{}}))=
\frac{2^{\windowpar{}}\alpha \insertnumber{\prob}}{\arraysize{\prob}}$ and each rebuild has cost $\arraysize{\prob}$
so the total  cost is at most $2^{\windowpar{}}\alpha \insertnumber{\prob}$.

Recalling that, for $k \in [1, \imax]$, we have $\Pr[\windowpar{} = k] = p_k = 2^{-(k + 1)} (1 + k / \imax) \le 2^{-k}$, we can take the expected value over all choices for $\windowpar{}$ to bound the expected total  cost of $\prob$'s rebuilds by
\[
\sum_{k=1}^{\imax} p_k 2^k\alpha \insertnumber{\prob}\leq 
\sum_{k=1}^{\imax} 2^{-k}2^k\alpha \insertnumber{\prob}= \imax \alpha \insertnumber{\prob} =O((\loglog n)^3\insertnumber{\prob}).
\]

Since the insertion sets for the subproblems at any fixed level of recursion are
disjoint, the total expected cost of rebuilds at each level
is $O(n(\log\log{n})^3)$.  Summing over the  $O(\log n)$ levels
yields an amortized expected cost of $O((\log n) (\loglog{n})^3)$ per operation.

Finally, because tiny subproblems have size at most $2^{O(\sqrt{\log n})}$, they incur amortized expected cost at most $O(\log n)$ per insertion.
\end{proof}

\subsection{Bounding the Costs of Expensive Leaves}

It remains to bound the cost of 
\failure leaves. These leaves may incur amortized cost as large as $O(\log^2 n)$ per insertion. So we want to show that the expected number of insertions that reach \failure leaves is $O(n/\log n)$. This is indeed true, and the proof occupies most of the rest of the paper.

Each insertion follows a unique root-to-leaf path in the history tree. We now define some notation for how to think about this path for a specific insertion $v$:
\begin{itemize}
    \item $\pathlength(v)$  is the number of subproblems in $v$'s path.
    \item  $\prob_1(v),\ldots,\prob_{\pathlength(v)}(v)$ is the path
of subproblems that $v$ follows.
\item $\freestart{j}(v)=\freestart{\prob_j(v)}$.
\item  $\free{j}(v)=\free{\prob_j(v)}(v)$.
\item For $\prob = \prob_j(v)$, $j < d(v)$, define $\freediff{\prob}(v)=\free{j+1}(v)-\free{j}(v)$, so that $\free{j}(v)-\free{1}(v)=\sum_{i=1}^{j-1} \freediff{\prob_i(v)}(v)$. 
\end{itemize}

\vspace{.2 cm}

All of the above are random variables that depend on both the sequence of insertions that has occurred prior to $v$ and the random choices of the algorithm, i.e., the parameters $\windowpar{\otherprob}$ for all subproblems $\otherprob$.

By definition,
the leaf subproblem $\prob_{\pathlength(v)}(v)$ is an \failure leaf if and only if $\freestart{\pathlength(v)}(v)\leq 1/4$ which implies $\free{\pathlength(v)}(v) \leq 1/4$.  On the other hand, since $\prob_1(v)$ is a global subproblem
and the total number of elements ever present is at most $m/2$,  $\free{\prob_1(v)} \geq 1/2$.
Thus, we obtain the following necessary condition for $v$ to reach an \failure leaf:

\begin{equation}
\label{eqn:failure1}
\sum_{i=1}^{\pathlength(v)-1}\freediff{\prob_i(v)}(v)  \leq -1/4.
\end{equation}

For a non-leaf subproblem $\prob$ and
insertion $v \in \inserts{\prob,i}$ (recall that $\inserts{\prob,i}$ denotes the insertions in the $i$-th rebuild window of $\prob$), define:

\[\Delta_{\prob}(v)= \begin{dcases}
     \frac{\iskewfinal{\prob}-2(1-\free{\prob})\arrayskew{i}}{\arraysize{\prob}-2\arrayskew{i}} & \text{if $v$ is sent left by $\prob$}\\
     \frac{-\iskewfinal{\prob}+2(1-\free{\prob})\arrayskew{i}}{\arraysize{\prob}+2\arrayskew{i}} & \text{if $v$ is sent right by $\prob$.}
     \end{dcases}
     \]
This definition comes out of the following
lemma, which shows that $\Delta_{\prob}(v)$ lower bounds $\freediff{\prob}(v)$.

 \begin{lemma}
 Let $\prob$ be a non-leaf subproblem and let $v \in \inserts{\prob}$. Then,
 \[
 \freediff{\prob}(v) \geq \Delta_{\prob}(v).
 \]
 
     \label{lem:delta}
 \end{lemma}
 \begin{proof}
Suppose $v$ occurs during the $j$-th rebuild window of $\prob$. 
Let $\otherprob$ be the child subproblem of $\prob$ (active during $\inserts{\prob,j}$) that $v$ is assigned to.  
 We will assume that $\otherprob$ is a left subproblem of $\prob$; the other case follows by a symmetric argument with the appropriate changes of sign.

 By definition, $\arraysize{\otherprob} = \arraysize{\prob}/2 - \arrayskew{\prob,j}$.   At the beginning
 of the first rebuild window of $\prob$, the left child of $\prob$ starts with $s_{\prob}/2$ items, where $s_{\prob}$ is the number of items
 initially stored in $\arr{\prob}$, and these items will be assigned to every left subproblem created in subsequent windows of $\prob$. The total number of inserts received by $\prob$ that go left is $(|\inserts{\prob}|-\iskewfinal{\prob})/2$, so the total number of elements ever stored in $\rho$ is at most
 $$s_{\prob}/2 + (|\inserts{\prob}|-\iskewfinal{\prob})/2,$$
 which implies that the total number of free slots in $\arr{\otherprob}$ is always at least
    \begin{eqnarray*}
    \arraysize{\otherprob} - s_{\prob}/2 - 
(|\inserts{\prob}|-\iskewfinal{\prob})/2 
    &=& 
    \arraysize{\prob}/2 - \arrayskew{\prob,j} - s_{\prob}/2 - 
(|\inserts{\prob}|-\iskewfinal{\prob})/2 \\
    &=& (\arraysize{\prob}-s_{\prob}-|\inserts{\prob}|)/2  +\iskewfinal{\prob}/2  - \arrayskew{\prob,j} \\
    & \geq  & \free{\prob} \arraysize{\prob}/2 + \iskewfinal{\prob}/2 - \arrayskew{\prob,j}.
    \end{eqnarray*}
    The free-slot density in $\otherprob$ at the end of its lifetime  therefore satisfies
    \begin{eqnarray*}
\free{\otherprob}  &\ge &\frac{\free{\prob}\arraysize{\prob}/2 + \iskewfinal{\prob}/2 -  \arrayskew{\prob,j}}{\arraysize{\otherprob}} \\
    & = &  \frac{\free{\prob} (\arraysize{\prob}/2 -  \arrayskew{\prob,j}) + \iskewfinal{\prob}/2 - (1-\free{\prob}) \arrayskew{\prob,j}}{\arraysize{\prob}/2 - \arrayskew{\prob,j}} \\
   &   = &  \free{\prob} + \frac{\iskewfinal{\prob}/2 - (1-\free{\prob}) \arrayskew{\prob,j}}{\arraysize{\prob}/2 -  \arrayskew{\prob,j}} \\
   &   = &  \free{\prob} + \frac{\iskewfinal{\prob} - 2(1-\free{\prob}) \arrayskew{\prob,j}}{\arraysize{\prob} - 2 \arrayskew{\prob,j}} = \free{\prob} + \Delta_{\prob(v)}. 
\end{eqnarray*}
 \end{proof}

Before continuing, it is worth remarking on two features of $\Delta_\prob(v)$, for $v \in \inserts{\prob,i}$, that make it nice to work with (and that, at least in part, shape its definition). 
     
The first property is that, if $D_\prob$ and $\arrayskew{i}$ were both zero (which would, happen, for example, if the insertions in $\prob$ alternated evenly between $\prob$'s left and right children), then $\Delta_\prob(v)$ would also be zero. This means that one should think of $\Delta_\prob(v)$ as having a ``default'' value of zero, which is why later on (in Lemma \ref{lemma:main}), when we want to bound $\operatorname{Var}(\Delta_\prob(v))$, we will be able to get away with bounding $\E[(\Delta_\prob(v))^2]$ instead.\footnote{Note that, no matter what, we have $\operatorname{Var}(\Delta_\prob(v)) \le \E[(\Delta_\prob(v))^2]$, so one can always use the latter as an upper bound for the former. What is important here, is that the latter quantity is actually a \emph{good} upper bound for the former.}

The second property is that $\Delta_\prob(v)$ is the \emph{same} for all $v$ in a given rebuild window $\inserts{\prob, i}$. In fact, if $\prob = \prob_j(v)$ for some $j$, then all $u \in \inserts{\prob}$ agree on the values of $\Delta_{\prob_1(u)}(u), \ldots, \Delta_{\prob_{j - 1}(u)}(u)$. This property will be critical for our analysis (in Lemma \ref{lemma:failure}), and later on, of how the sequence $\Delta_{\prob_1(v)}(v), \Delta_{\prob_2(v)}(v), \ldots$ behaves. This property is also the reason why all of the quantities used to define $\Delta_\prob(v)$ (i.e., $D_\prob(v), F_j(v),\arrayskew{\prob, i}$) are based only on the window $\inserts{\prob, i}$ that contains $v$, rather than on anything more specific about the insertion $v$.

For an insert $v$, define $\Delta_1(v),\ldots,\Delta_{\pathlength(v)-1}(v)$
by  $\Delta_i(v)=\Delta_{\prob_i(v)}(v)$ for $i<\pathlength(v)$. 
Combining Lemma \ref{lem:delta} with  \eqref{eqn:failure1}, we get a new necessary condition for $v$ to reach an \failure leaf:

\[
\sum_{i=1}^{\pathlength(v)-1} \Delta_i(v)\leq -1/4.
\]

We will bound the fraction of $v$'s that arrive at an \failure leaf by showing that at most an expected $O(1/\log n)$ fraction of insertions $v$ satisfy the above condition.  To analyze this fraction,
we  adopt a probabilistic point of view with regard
to the insertions themselves. In particular, rather than analyzing the probability that any \emph{specific} insertion reaches an \failure leaf, we will select a uniformly random insertion $v$ from the entire insertion sequence (this randomness is for the sake of analysis, only, and is not coming from the randomness in the algorithm), and we will analyze the probability that this randomly selected insertion reaches an \failure leaf. Thus, the underlying
probability space will be over both
randomly chosen $v$ and the randomness of the algorithm.

To analyze a random insertion $v$, our main task will be to analyze the (random) sequence $\Delta_{1}(v),\Delta_{2}(v),\ldots$. We first make the observation
(Proposition~\ref{prop:Delta UB})  that each $\Delta_i(v)$ is
bounded in absolute value by $\frac{3}{\beta}$.  The more significant
result is then Lemma~\ref{lemma:main}, which says that under fairly general conditions on a subproblem $\prob$, the variance of
$\Delta_{\prob}(v)$, for a uniformly random $v \in \inserts{\prob}$, is (deterministically) bounded above by a small (sub-constant) multiple of its expectation (up to a negligible additive term). Once we have these properties, we will argue that they force $\Delta_1(v), \Delta_2(v), \ldots$ to evolve according to a well-behaved process, which we will then analyze using (mostly) standard results from the theory of random walks.

 \begin{proposition}
     \label{prop:Delta UB}
    For any non-leaf subproblem $\prob$ and any $v \in \inserts{\prob}$, we have $|\Delta_\prob(v)| \leq \frac{3}{\beta}$.
 \end{proposition}

\begin{proof}
Let $\inserts{i}$ be the rebuild window of $\inserts{\prob}$ such that $v \in \inserts{i}$. Note that the array skew $\arrayskew{i}$ satisfies $|\arrayskew{i}| \le  \arraysize{\prob} \cdot \frac{|\iskewfinal{\prob,i - 1}|}{\beta \window{\prob}} \le \arraysize{\prob}/\beta$. Using this, we can conclude that

\begin{align*}
|\Delta_{\prob}(v)| & \leq \frac{|\iskewfinal{\prob}|+2|\arrayskew{i}|}{\arraysize{\prob} - 2|\arrayskew{i}|} \\ & \le (1 + 3/\beta) \cdot \left(
\frac{|\iskewfinal{\prob}|+2|\arrayskew{i}|}{\arraysize{\prob}}\right) \\
& \le (1 + 3/\beta) \cdot \left(
\frac{|\iskewfinal{\prob}|}{\arraysize{\prob}}\right) + (1 + 3/\beta) \cdot 2 / \beta \\
& \le (1 + 3/\beta)/\alpha + (1 + 3/\beta) \cdot 2 / \beta \tag{since $|\iskewfinal{\prob}| \le |\inserts{\prob}| \le \arraysize{\prob} / \alpha$} \\
& \le 3/\beta,
\end{align*}
where the final inequality uses that $\beta = \omega(1)$ and $\alpha > 100 \beta$.
\end{proof}
We now come to the  main technical lemma, which we will prove in Section~\ref{sec:main}. 

\begin{restatable}{lemma}{mainlemma}\hspace{-.1 cm}(The See-Saw Lemma) \hspace{.1 cm}
\label{lemma:main}
Let $\prob$ be a non-leaf subproblem with insertion set $\inserts{\prob}$, and suppose that $\freestart{\prob}$, the free-slot density of $\prob$ when it starts,
satisfies $\freestart{\prob} \le 0.5$. Then, for a uniformly random $v \in \inserts{\prob}$,
 $$\E[\Delta_\prob(v)^2]   \leq  \frac{100 \imax}{\beta} \E[\Delta_\prob(v)] + 2^{-\imax},$$
    where the expectations are taken over both the random choice of $v$ and the algorithm's random choice of $\window{\prob}$.
\end{restatable}

Lemma \ref{lemma:main} is the part of the analysis that captures the role of \emph{adaptivity} in our algorithm. If the algorithm were not adaptive (i.e., always set $\arrayskew{i} = 0$), then $\Delta_\prob(v)$ would simply be $D_\prob / m_\prob$. The insertion sequence would then be able to force $\Delta_\prob(v)^2 = (D_\prob / m_\prob)^2$ (and, more importantly, $\operatorname{Var}(\Delta_\prob(v))$) to be large by sending more insertions to one child of $\prob$ than to the other. (This would also cause $\E[\Delta_\prob(v)]$ to be slightly negative, which would also be bad for us.) The key insight in Lemma \ref{lemma:main} is that we cannot hope to prevent $\Delta_{\prob}(v)^2$ from being large---but we \emph{can hope} to use adaptivity in order to create a ``see-saw'' relationship between $\E[\Delta_{\prob}(v)^2]$ and $\E[\Delta_{\prob}(v)]$. In particular, if the insertion sequence chooses to send far more insertions to one child than the other, then this creates an opportunity for us to employ adaptivity, which we can then use to put more free slots on the side that receives more insertions, which allows for us to create a positive expected value for $\E[\Delta(v)]$. This, in turn, is a good thing, since $\Delta(v)$ being positive means that, on average, insertions experience a free-slot density \emph{increase} when traveling from $\prob$ to $\prob$'s child. Thus, we create a situation where, no matter what, we win: either $\E[\Delta(v)]$ is large (which is good), or $\operatorname{Var}(\Delta_\prob(v))$ is small (which is also good!). 

Thus, the ``magic'' of the See-Saw Algorithm will be in how it uses adaptivity to guarantee the See-Saw Lemma. A priori, the adaptive behavior of the algorithm (i.e, the way in which it selects $\arrayskew{\prob, i}$ based on the insertion behavior in the previous rebuild window) would seem to be quite difficult to analyze. Intuitively, the algorithm is attempting to observe when there are ``trends'' in the insertion-sequence's behavior. However, if we are not careful, the insertion sequence may be able to trick us into observing a ``trend'' in one rebuild window, even though the next rebuild window will behave in the opposite way. The main contribution of Section \ref{sec:main}, where we prove Lemma \ref{lemma:main}, is that if the window size $\window{\prob}$ and the array skews $\arrayskew{\prob, 1}, \arrayskew{\prob, 2}, \ldots$ are selected in just the right way (as in the See-Saw Algorithm), then it is possible to perform a telescoping argument that holds for \emph{any input}. The argument shows that, even if the insertion sequence causes the algorithm to perform badly for some choices of $\window{\prob}$, this creates ``opportunities'' for the algorithm to perform better on other choices of $\window{\prob}$, so that on average the algorithm always does well.

Since Lemma \ref{lemma:main} only considers subproblems $\prob$ satisfying $\freestart{\prob} \le 0.5$, it will be useful to  define a modified version of $\Delta_\prob$, where for any non-leaf subproblem $\prob$, we have:

\[
\Deltahat_{\prob}(v)=\begin{cases}
\Delta_{\prob}(v) & \text{if }\freestart{\prob} \le 0.5\\
0 & \text{ otherwise}.
\end{cases}
\]

Trivially, we then have:
\begin{corollary}
\label{cor:Deltahat}
For any non-leaf subproblem $\prob$, and for a uniformly random $v \in \inserts{\prob}$, we have
$$\E[\Deltahat_\prob(v)^2]   \leq  \frac{100 \imax}{\beta} \E[\Deltahat_\prob(v)] + 2^{-\imax}.$$
\end{corollary}

We define $\Deltahat_1(v),\Deltahat_2(v),\ldots$ by $\Deltahat_i(v)=\Deltahat_{\prob_i(v)}(v)$.  Earlier we gave a necessary condition for reaching an \failure leaf, based on summations
of $\{\Delta_i(v)\}$.  We now give a similar condition based on the sequence $\{\Deltahat_i(v)\}$:

\begin{proposition}
\label{prop:failure}
For any insert $v$, if  the path of $v$ ends at an \failure leaf, then
there is an interval $[a,b] \subseteq [1, d(v) - 1]$ such that 
$$\Deltahat_a(v)+\cdots+\Deltahat_b(v) \leq -0.23.$$
\end{proposition}

\begin{proof}
Suppose that the leaf $\prob_{\pathlength(v)}(v)$ is an \failure leaf. 
So $\freestart{\pathlength(v)}(v) \leq 1/4$, which
implies $\free{\pathlength(v)}(v) \leq 1/4$.
Let $b=d(v)-1$ and let $\ell$ be the largest index such that
$\freestart{\ell}(v) \ge 1/2$; this is well-defined because $\prob_1(v)$ is global and so $\freestart{1}(v) \geq 1/2$.
Since $\alpha \ge 100$, and since $\prob_\ell(v)$ gets at most $\arraysize{\prob_{\ell}(v)} / \alpha$ insertions, we have $\free{\ell}(v) \ge \freestart{\ell}(v)-0.01 \geq 0.49$.
By the definition of $\ell$, 
$\Deltahat_i(v)=\Delta_i(v)$ for all  $i \in [\ell+1,b]$.
Letting $a=\ell+1$, and letting $\delta_i(v)$ denote $\delta_{\prob_i(v)}$,  we have by Lemma \ref{lem:delta} that

$$\sum_{i=a}^b \Deltahat_i(v) = - \Delta_\ell(v) + \sum_{i=\ell}^b \Delta_i(v) \le - \Delta_\ell(v) + \sum_{i=\ell}^b \delta_i(v) =   - \Delta_\ell(v) + \free{b+1}(v)-\free{\ell}(v) \leq - \Delta_\ell(v) + 1/4 - 0.49.$$
Finally, by Proposition~\ref{prop:Delta UB}, this is at most $3/\beta + 1/4-0.49 \leq 0.01 + 1/4-0.49$, since $\beta \ge 300$.
\end{proof}

We will now show how to bound the probability that a random insertion $v$ (out of the \emph{entire} input stream) encounters an \failure leaf.

\begin{lemma}
\label{lemma:failure}
    Consider a uniformly random insertion $v$ out the entire insertion stream. The probability that $v$ reaches an \failure leaf is $O(1 /\log n)$.
\end{lemma}
\begin{proof}
Since $v$ is a random variable, the sequences $\prob_1(v), \prob_2(v), \ldots$ and $\Deltahat_{1}(v), \Deltahat_2(v), \ldots$ are also random variables. When discussing our randomly chosen $v$, we will use $\Deltahat_j$ as a shorthand for $\Deltahat_j(v)$.
For convenience of notation, we let $\Deltahat_j = 0$ if $j \geq \pathlength(v)$.

By Proposition~\ref{prop:failure}, it suffices to upper bound the probability
that there is a pair $a \leq b$ satisfying $\Deltahat_a+\cdots+\Deltahat_b \leq -.23$.
The
maximum depth of the history tree is $2\log n$ so there are at most
$4\log^2n$ pairs with $a\leq b <2\log n$.  So it suffices to fix $a \leq b < 2\log n$ and show that $\Pr[\sum_{i=a}^b \Deltahat_{i}(v)\leq -0.23] =
O(1/\log^3 n)$.

In Section~\ref{sec:claim} we will use standard concentration bounds for random processes to show:

\begin{restatable}{claim}{concclaim}
Let $\imax = 2 \log \log n$, and let $\beta = \cbeta (\log \log n)^2$ for some sufficiently large positive constant $\cbeta$. Let $X_1, X_2, \ldots,X_r$ be random variables with $r\leq 2 \log n$
such that for $i \in [1,r]$:
\begin{enumerate}
    \item $|X_i| \le 3/\beta$;
    \item $\E[X_i^2\mid X_1, \ldots, X_{i - 1}]   \leq  \frac{100 \imax}{\beta} \E[X_i\mid X_1, \ldots, X_{i - 1}] + 2^{-\imax}.$
\end{enumerate}
Then,  $\Pr[\sum_i X_i < -0.2] \le O\left(1 / \log^3 n\right)$.
\label{claim:conc}
\end{restatable}

Although we will defer the proof of Claim \ref{claim:conc} to Section \ref{sec:claim}, it may be worth taking a moment to explain the intuition behind the claim. For this, it is helpful to substitute $X_i$ with $X_i' := X_i \cdot \log \log n$. Under this substitution (and with a bit of algebra) one can reduce the the hypotheses of the claim to  (1) $|X_i'| \le O(1 / \log \log n) \le 1$, and (2) $\E[{X'_i}^2 \mid X'_1, \ldots, X'_{i - 1}] \le O(1) \cdot \E[X'_i \mid X'_1, \ldots, X'_{i - 1}] + \tilde{O}(1/\log^2 n)$; and the conclusion of the claim becomes that, with probability $1 - 1 / \log^3 n$, we have $\sum_i X'_i \ge -O(\log \log n)$. In other words, the essence of the claim is simply that, if a random walk has steps of size at most, say $1$, and if each step has mean at least a constant factor larger than its variance (modulo some small additive error), then the random walk will not be able to become substantially negative with any substantial probability.

We would like to apply Claim \ref{claim:conc} to $X_1, X_2, \ldots = \Deltahat_a, \Deltahat_{a + 1}, \ldots, \Deltahat_{b}$. Proposition~\ref{prop:Delta UB} implies that each $\Deltahat_i$ satisfies the first
hypothesis of the claim, and the second hypothesis \emph{almost} follows from Corollary~\ref{cor:Deltahat}. The only issue is that Corollary~\ref{cor:Deltahat} tells us how to think about $\Deltahat_\prob(v)$ for a random $v$ \emph{out of those in $\inserts{\prob}$}, but what we actually want to reason about is $\Deltahat_i(v) \mid \Deltahat_a(v), \ldots, \Deltahat_{i - 1}(v)$ for a random $v$ \emph{out of all insertions}. Fortunately, these two probability distributions end up (by design) being  closely related to one another, allowing us to establish the following variation of Corollary \ref{cor:Deltahat}:

\begin{corollary}
\label{cor:Deltahat2}
For each $i \in [1,2\log n]$,
    $$\E[\Deltahat_i^2\mid \Deltahat_1,\ldots,\Deltahat_{i-1}]   \leq  \frac{100 \imax}{\beta} \E[\Deltahat_i \mid \Deltahat_1,\ldots,\Deltahat_{i-1}] + 2^{-\imax}.$$
\end{corollary}

\begin{proof}
Because, in this proof, we will use $\prob_i$ as a formal random variable, it is helpful to think of each $\prob_i$ as formally being given by the triple $(\arr{\prob_i}, \set{\prob_i}, \inserts{\prob_i})$. 

Recall that our probability space consists of the selection of the parameters $\window{\prob}$ during the algorithm (which, along with the insertion sequence, fully determine the history tree) and
the selection of a uniformly random $v$ that determines the path  $\prob_1,\prob_2,\ldots,\prob_j$ down the tree. We will need  an alternative incremental description of the probability space.
Keep in mind that the full sequence of insertions is fixed.
First note that the global subproblems are completely determined by the insertion sequence (i.e., there is
no randomness) and the insertion sets for these subproblems partition the
full set of insertions. Select $\prob_1$ from among the global
subproblems with probability proportional to the size of
its set of insertions. Next select the parameter $\window{1}=\window{\prob_1}$ according to the algorithm specification.
The parameter $\window{1}$ determines 
the windows and the set of subproblems of $\prob_1$, and the insertion sets
for these subproblems partition the insertion set of $\prob_1$. Next we select $\prob_2$ 
from among these subproblems with probability proportional to the size of
its insert set.  We continue in this way selecting $\prob_1,\window{1},\prob_2,\window{2},\ldots$ until we arrive either at a tiny leaf
or \failure leaf. This process
gives the same distribution over paths $\prob_1, \prob_2, \ldots,$ as the distribution 
that first runs the algorithm to determine the full 
history tree and then selects a random
insert and follows its path.

Now let us consider the random variable $\Deltahat_i \mid \prob_1, w_1, \ldots, \prob_i$, for a given $i \in [1, 2 \log n]$. So that this is well defined for all $i$, we can artificially define $\prob_j(v)$ and $w_j(v)$ to be null, for $j > d(v)$ and $j \ge d(v)$, respectively. If $\prob_i$ is a leaf (or null), then $\Deltahat_i \mid \prob_1, w_1, \ldots, \prob_i$ is defined to be identically zero, so we have trivially that 
$$\E[\Deltahat_i^2\mid \prob_1,\window{1},\ldots,\prob_i]   \leq  \frac{100 \imax}{\beta} \E[\Deltahat_i \mid\prob_1,\window{1},\ldots,\prob_i] + 2^{-\imax}.$$
The interesting case is what happens if $\prob_i$ is a non-leaf subproblem.

For any given set of outcomes for $\prob_1,\window{1},\ldots,\prob_i$, where $\prob_i$ is a non-leaf subproblem, the probabilistic rule for selecting  $\window{i}$ and $\prob_{i+1}$ (which together determine $\Deltahat_i$) is completely determined by $\prob_i$. Therefore, if we fix any set of outcomes for $\prob_1,\window{1},\ldots,\prob_i$, and if we use $\Deltahat_{\prob}$  to denote the random variable $\Deltahat_{\prob}(u)$ for a uniformly random $u \in \inserts{\prob}$, then we have

\begin{eqnarray*}
    \E[\Deltahat_i \mid \prob_1,\window{1},\ldots,\prob_i] = \E[\Deltahat_i \mid \prob_i] = \E[\Deltahat_{\prob_i}],\\
    \E[\Deltahat_i^2 \mid \prob_1,\window{1},\ldots,\prob_i] =  \E[\Deltahat_i^2 \mid \prob_i] =  \E[\Deltahat_{\prob_i}^2].\\
\end{eqnarray*}

Now we can combine this with Corollary~\ref{cor:Deltahat} to obtain
$$\E[\Deltahat_i^2\mid \prob_1,\window{1},\ldots,\prob_i]   \leq  \frac{100 \imax}{\beta} \E[\Deltahat_i \mid\prob_1,\window{1},\ldots,\prob_i] + 2^{-\imax}.$$
Since, earlier in the proof, we also established this identity for the case where $\prob_i$ is a leaf (or null), we can conclude that the identity holds for all options of $\prob_1, \window{1}, \ldots, \prob_i$.

Finally note that $\prob_1,\window{1},\ldots,\prob_i$ determines
$\Deltahat_1,\ldots,\Deltahat_{i-1}$.  Therefore for any fixing $d_1,\ldots,d_{i-1}$
of $\Deltahat_1,\ldots,\Deltahat_{i-1}$, we can average the previous inequality
with respect to the conditional distribution
on $\prob_1,\window{1},\ldots,\prob_i$ given $\Deltahat_1=d_1,\ldots,\Deltahat_{i-1}=d_{i-1}$ and this gives
exactly the desired result.
\end{proof}

An immediate consequence of Corollary \ref{cor:Deltahat2} is that, for any interval $[a, b]$, and for $i \in [a, b]$, we have $\E[\Deltahat_i^2 \mid \Deltahat_a, \ldots, \Deltahat_{i - 1}] \le \frac{100\imax}{\beta}  \E[\Deltahat_i \mid \Deltahat_a, \ldots, \Deltahat_{i - 1}]  + 2^{-\imax}$. We also have by Proposition  \ref{prop:Delta UB} that $|\Deltahat_i| \le 3/\beta$, so we can apply Claim~\ref{claim:conc}, using $X_1, X_2, \ldots = \Deltahat_a, \Deltahat_{a + 1}, \ldots, \Deltahat_{b}$ to complete the proof that $\Pr[\Deltahat_a + \cdots + \Deltahat_b \le -0.23] = O(1 / \log^3 n)$, which, in turn, completes the proof of the lemma. 
\end{proof}

Given Lemma \ref{lemma:failure}, we can complete the proof of Theorem \ref{thm:nicetoprove} as follows.

\thmnice*
\begin{proof}
    Lemma \ref{lem:simplecosts} bounds the amortized expected costs of rebuilds, resets, and tiny leaves by $$O((\log n) (\log \log n)^3).$$ Lemma \ref{lemma:failure} bounds the probability of an insertion encountering an \failure leaf by $O(1 / \log n)$. If an insertion does encounter an \failure leaf, it incurs $O(\log^2 n)$ amortized expected cost within the leaf. Thus, the amortized expected cost per insertion from \failure leaves is $O(\log n)$. 
\end{proof}

It remains to prove Lemma~\ref{lemma:main} and Claim~\ref{claim:conc}.
These are given in Sections~\ref{sec:main} and~\ref{sec:claim}.

\section{Proof of The See-Saw Lemma}
\label{sec:main}

In this section, we prove Lemma \ref{lemma:main}, restated below:

\mainlemma*

Let us also recall the definition of $\Delta_{\prob}(v)$. For a subproblem $\prob$ and
insertion $v \in \inserts{\prob,j}$ (recall that $\inserts{\prob,j}$ 
denotes the insertions in the $j$-th rebuild window of $\prob$):

\[\Delta_{\prob}(v)= \begin{dcases}
     \frac{\iskewfinal{\prob}-2(1-\free{\prob})\arrayskew{j}}{\arraysize{\prob}-2\arrayskew{j}} & \text{if $v$ is sent left by $\prob$}\\
     \frac{-\iskewfinal{\prob}+2(1-\free{\prob})\arrayskew{j}}{\arraysize{\prob}+2\arrayskew{j}} & \text{if $v$ is sent right by $\prob$.}
     \end{dcases}
     \]
Since $\prob$ is the only subproblem
that will be mentioned in this proof, we'll often omit the subscript $\prob$.

 To prove the lemma, we will prove a lower bound on $\E[\Delta(v)]$ and
 an upper bound on $\E[\Delta(v)^2]$ and then compare them.
 These expectations would be easier to deal with if we could change
 the denominator in $\Delta(v)$ to $\arraysize{\prob}$.  In particular, the
 expressions for insertions $v$ that go left vs right would then be negatives of each other. Recalling that $\iskewfinal{}(v)$ is $+1$ if $v$ goes right and $-1$ if $v$ goes left, define

\[
 \Lambda(v)=\Lambda_{\prob}(v)=\iskewfinal{}(v)\frac{-\iskewfinal{\prob}+2(1-\free{\prob})\arrayskew{j}}{\arraysize{\prob}}.
\]
We will use $\Lambda(v)$ to estimate $\Delta(v)$.  Define the error function 

\[
\varepsilon(v)=\Delta(v)-\Lambda(v).
\]

The following claim will inform how we think about  $\varepsilon(v)$:

\begin{claim}
If $v \in \inserts{\prob,j}$ for some $j$, then we have
\[\varepsilon(v) =\begin{dcases}\frac{2\arrayskew{j}}{\arraysize{\prob}-2\arrayskew{j}} \Lambda(v) & \text{if $v$ is sent left by $\prob$}\\
\frac{-2\arrayskew{j}}{\arraysize{\prob}+2\arrayskew{j}} \Lambda(v) & \text{if $v$ is sent right by $\prob$,}
     \end{dcases}
 \]
and that   $$|\varepsilon(v)| \leq \frac{8|\arrayskew{j}|}{3\arraysize{\prob}}|\Lambda(v)| \le |\Lambda(v)|/3.$$
\label{clm:lambda}
 \end{claim}
 \begin{proof}
    We have $\frac{\Delta(v)}{\Lambda(v)}  = \frac{m_\prob}{m_\prob + D(v) \cdot 2\arrayskew{j}}$, which implies
    \begin{align*}
        \frac{\Delta(v) -\Lambda(v) }{\Lambda(v)} & =   \frac{m_\prob - (m_\prob + D(v) \cdot 2\arrayskew{j}) }{m_\prob + D(v) \cdot 2\arrayskew{j}} = \frac{- D(v) \cdot 2\arrayskew{j}}{m_\prob + D(v) \cdot 2\arrayskew{j}}.
    \end{align*}
    To prove the second part of the claim, observe that  $|\arrayskew{j}|=\frac{|\iskewfinal{j-1}|\arraysize{\prob}}{\beta\window{j}} \leq \frac{\arraysize{\prob}}{\beta}\leq \frac{\arraysize{\prob}}{8}$, and so:
\begin{equation*}
    |\varepsilon(v)| \leq \frac{2|\arrayskew{j}|}{\frac{3}{4}\arraysize{\prob}}|\Lambda(v)| = \frac{8|\arrayskew{j}|}{3\arraysize{\prob}}|\Lambda(v)| \le |\Lambda(v)|/3.
\end{equation*}
\end{proof}

We can bound $\E[\Delta(v)]$ and $\E[\Delta(v)^2]$ as a function of $\Lambda(v)$ and $\varepsilon(v)$ as follows:
\begin{proposition}
\label{prop:conversion}
For a subproblem $\prob$,
\begin{eqnarray*}
    \E[\Delta(v)]& \geq &\E[\Lambda(v)] - \E[|\varepsilon(v)|]\\
    \E[\Delta(v)^2] &\leq & 2 \cdot \E[\Lambda(v)^2], 
\end{eqnarray*}
where expectations are with respect to the randomness of the algorithm and $v$ chosen uniformly from $\inserts{\prob}$.
\end{proposition}

\begin{proof}
    The first inequality is immediate from the definition of $\varepsilon(v)$ and the triangle inequality.  For the second, using the bound $|\varepsilon(v)|\leq |\Lambda(v)|/3$ from Claim \ref{clm:lambda}, we have:
    \[
    \E[\Delta(v)^2]=\E[(\Lambda(v)+\varepsilon(v))^2] \leq \E[(|\Lambda(v)|+|\Lambda(v)|/3)^2] \leq \frac{16}{9}\E[\Lambda(v)^2].
    \]
    
    \end{proof}

In what follows, we will  compute a lower bound on $\E[\Lambda(v)]$ and
upper bounds on $\E[\Lambda(v)^2]$ and $\E[|\varepsilon(v)|]$.  We will then be able to use Proposition~\ref{prop:conversion} to complete the proof of Lemma \ref{lemma:main}.

Recall that, for the subproblem $\prob$, the algorithm chooses its rebuild window size $\window{\prob}$ based on the random variable $\windowpar{\prob}\in [0,\imax]$. It will often be helpful to condition on $\windowpar{\prob}=k$ for some $k$. Thus we use the following notation to refer to the values that variables take when $\windowpar{\prob} = k$:
\begin{itemize}
\item $\window{}^k$ is the window size $\arraysize{\prob}/(\alpha 2^k)$.
\item $\numwin{}^k$ is the number of rebuild windows. 
\item The partition of $\inserts{\prob}$ into windows is denoted $\inserts{1}^k,\inserts{2}^k,\ldots,\inserts{\numwin{}^k}^k$. \\(We have $\numwin{}^k \leq 2^k$, since $\frac{\arraysize{\prob}}{\alpha}\geq |\inserts{\prob}| =\sum_j|\inserts{j}^k| >(\numwin{}^k-1) \window{}^k \geq  (\numwin{}^k-1) \frac{\arraysize{\prob}}{\alpha 2^k}$.)
\item $\iskewfinal{j}^k$ is an abbreviation for $\iskewfinal{}(\inserts{j}^k)$.
\item $\arrayskew{j}^k$ is the value used by the algorithm for
$\arrayskew{j}$.  It is $\frac{\arraysize{\prob}\iskewfinal{j-1}^k}{\beta\window{}^k}$
if $j$ is even, and is 0 if $j$ is odd.
\end{itemize}

\vspace{.2 cm}

Observe that the rebuild windows for $\windowpar{\prob} =k+1$ are obtained
by splitting each rebuild window for  $\windowpar{\prob}=k$ into two parts $\inserts{j}^k = \inserts{2j-1}^{k+1} \cup \inserts{2j}^{k+1}$. The two sets $\inserts{2j-1}^{k+1}$ and $\inserts{2j}^{k+1}$ will both be of size $\window{}^{k+1}$, unless $j=\numwin{}^k$, in which case the sizes may be less than $\window{}^{k+1}$ or even $0$; indeed the rebuild window $\inserts{2\numwin{}^k}^{k+1}$ may not even exist, in which case we treat it as empty.

The following two sums play a key role in the computation of $\E[\Lambda(v)]$ and $\E[\Lambda(v)^2]$.
\begin{eqnarray*}
S^k&=&\sum_{j\leq t^k} (\iskewfinal{j}^k)^2.\\
R^k&=&\sum_{\text{even } j \leq \numwin{}^k}\iskewfinal{j-1}^k\iskewfinal{j}^k. 
\end{eqnarray*}
In the case that there is only one window (e.g. $k=0$), we have $S^k = (\iskewfinal{\prob})^2$ and $R^k = 0$.

The following upper bound on $S^k$ will be helpful later on in the proof, in particular, when we wish to bound $S^{\imax}$. 
Recall, $\insertnumber{\prob}= |\inserts{\prob}|$.
\begin{proposition}
    \label{prop:S^h}
    For any $k \in [1,\imax]$,
    $$S^k \leq \frac{\arraysize{\prob}\insertnumber{\prob}}{\alpha 2^k}.$$
\end{proposition}
\begin{proof}
We have that $$S^k=\sum_{j=1}^{\numwin{}^k} (\iskewfinal{j}^k)^2 \leq \sum_{j=1}^{\numwin{}^k} |\inserts{j}^k|^2 \leq 
\sum_{j=1}^{\numwin{}^k} |\inserts{j}^k| \frac{\arraysize{\prob}}{\alpha 2^k} = \frac{\insertnumber{\prob} \arraysize{\prob}}{\alpha 2^k}.$$
\end{proof}

We now turn to our bounds on $\E[\Lambda(v)]$,
$\E[(\Lambda(v))^2]$, and $\E[|\varepsilon(v)|]$.

\begin{lemma}
\label{lemma:expectation}
For any subproblem $\prob$, we have: 
    \begin{eqnarray*}
        \E[\Lambda(v)]& = &\frac{\alpha(1-\free{\prob})}{\arraysize{\prob}\insertnumber{\prob} \beta} \sum_{k=1}^{\imax} \left(1+\frac{k}{\imax}\right) R^k - \frac{S^0}{\insertnumber{\prob}\arraysize{\prob}}\\
    \E[(\Lambda(v))^2] & \leq & \frac{8\alpha }{\arraysize{\prob}\insertnumber{\prob}\beta^2} \sum_{k=0}^{\imax}  S^k \\
     \E[|\varepsilon(v)|] & \leq & \frac{8\alpha }{\arraysize{\prob}\insertnumber{\prob}\beta^2} \sum_{k=0}^{\imax}  S^k, 
    \end{eqnarray*}
    where the expectations are taken with respect to random $v \in \inserts{\prob}$ and the choice of $\windowpar{\prob}$.
\end{lemma}
The proof of this lemma follows from straightforward calculations:
\begin{proof}

To bound $\E[\Lambda(v)]$, we first analyze $\E[\Lambda(v)]$ conditioned
on $\windowpar{\prob}=k$.  
We write $\Lambda^k(v)$ (resp.~$\varepsilon^k(v)$) for $\Lambda(v)$ (resp.~$\varepsilon(v)$) conditioned on  $\windowpar{\prob}=k$.  After this conditioning, the only remaining randomness is the uniform random choice of $v \in \inserts{\prob}$. For each window $\inserts{j}^k$, and for all $v \in \inserts{j}^k$, we have by definition that
    $\Lambda^k(v)=\iskewfinal{\prob}(v) \frac{2(1-\free{\prob})\arrayskew{j}^k-\iskewfinal{\prob}}{\arraysize{\prob}}$,
    so: 
    \[
    \sum_{v \in \inserts{j}^k} \Lambda^k(v)=\iskewfinal{j}^k\frac{2(1-\free{\prob})\arrayskew{j}^k-\iskewfinal{\prob}}{\arraysize{\prob}}.
    \]
    Therefore, for $v$ selected uniformly at random from $\inserts{\prob}$, we have
    \begin{align*}
    \E[\Lambda^k(v)] & =  \frac{1}{\insertnumber{\prob}}\sum_{j=1}^{\numwin{}^k} \iskewfinal{j}^k\frac{2(1-\free{\prob})\arrayskew{j}^k-\iskewfinal{\prob}}{\arraysize{\prob}}\\
    & = \frac{2(1-\free{\prob})}{\insertnumber{\prob}\arraysize{\prob}}\sum_{j=1}^{\numwin{}^k}\iskewfinal{j}^k\arrayskew{j}^k - \frac{(\iskewfinal{\prob})^2}{\insertnumber{\prob}\arraysize{\prob}} \tag{since  $\sum_{j=1}^{\numwin{}^k}\iskewfinal{j}^k=\iskewfinal{\prob}$} \\  
     & = \frac{2(1-\free{\prob})}{\insertnumber{\prob}\arraysize{\prob}}\sum_{\text{even }j \leq \numwin{}^k}\frac{\arraysize{\prob}}{\beta\window{\prob}^k}\iskewfinal{j-1}^k\iskewfinal{j}^{k} - \frac{(\iskewfinal{\prob})^2}{\insertnumber{\prob}\arraysize{\prob}} \tag{by definition of $\arrayskew{j}^k$} \\ 
 & = \frac{(1-\free{\prob})2^{k+1}\alpha}{\insertnumber{\prob}\arraysize{\prob} \beta} R^k - \frac{S^0}{\insertnumber{\prob}\arraysize{\prob}}.
    \tag{since $|\window{}^k|=\frac{\arraysize{\prob}}{\alpha 2^k}$ and $S^0=(\iskewfinal{\prob})^2$}\\
    \end{align*}

    Now averaging over the options for $k$, each of which occurs with probability $p_k = 2^{-(k + 1)} \cdot (1 + k/\imax)$,

    \begin{align*}
        \E[\Lambda(v)] & = \sum_{k=0}^{\imax} \frac{p_k (1-\free{\prob})2^{k+1}\alpha}{\insertnumber{\prob}\arraysize{\prob} \beta} R^k - \left(\sum_k p_k\right) \cdot \frac{S^0}{\insertnumber{\prob}\arraysize{\prob}}\\
        & =  \sum_{k=1}^{\imax} \frac{p_k (1-\free{\prob})2^{k+1}\alpha}{\insertnumber{\prob}\arraysize{\prob} \beta} R^k - \frac{S^0}{\insertnumber{\prob}\arraysize{\prob}} \tag{since $R^0=0$ and $\sum p_k = 1$}\\        
        & = \frac{\alpha(1-\free{\prob})}{\insertnumber{\prob}\arraysize{\prob}\beta} \sum_{k=1}^{\imax} \left(1+\frac{k}{\imax}\right) R^k - \frac{S^0}{\insertnumber{\prob}\arraysize{\prob}},  \\
    \end{align*}
    as claimed.

    Next, we analyze $\E[(\Lambda(v))^2]$. As above we start by analyzing the
    conditional expectation $\E[(\Lambda^k(v))^2]$.
    For each window $\inserts{j}^k$, and for each $v \in \inserts{j}^k$ we have: 
    \begin{align}
    \nonumber
    (\Lambda^k(v))^2 & = \frac{(2(1-\free{\prob})\arrayskew{j}^k-\iskewfinal{\prob})^2}{(\arraysize{\prob})^2}\\ 
    & \leq  2\frac{(2(1-\free{\prob})\arrayskew{j}^k)^2 +(\iskewfinal{\prob})^2}{(\arraysize{\prob})^2} \tag{by the inequality $(a + b)^2 \leq 2a^2 + 2b^2$} \\
    & \leq \frac{8(\arrayskew{j}^k)^2 +2(\iskewfinal{\prob})^2}{(\arraysize{\prob})^2}, 
    \label{eqn:Delta^2}
    \end{align}
    where the final step uses $F_\prob \in [0, 1]$. Therefore,
    \[
    \sum_{v \in \inserts{j}^k}(\Lambda^k(v))^2 \leq |\inserts{j}^k| \frac {8(\arrayskew{j}^k)^2+2(\iskewfinal{\prob})^2}{(\arraysize{\prob})^2}.
    \]
    It follows that, for a uniformly random $v \in \inserts{\prob}$, we have
    \begin{align*}
    \E[(\Lambda^k(v))^2]& =  \frac{1}{\insertnumber{\prob}}\sum_{j=1}^{\numwin{}^k}|\inserts{j}^k| \frac{8(\arrayskew{j}^k)^2+2(\iskewfinal{\prob})^2}{(\arraysize{\prob})^2} \\
    & =  \frac{8}{\insertnumber{\prob}(\arraysize{\prob})^2}\sum_{j=1}^{\numwin{}^k}|\inserts{j}^k| (\arrayskew{j}^k)^2 + \frac{2(\iskewfinal{\prob})^2}{\insertnumber{\prob}(\arraysize{\prob})^2}\sum_{j=1}^{\numwin{}^k}|\inserts{j}^k|\\
    & =  \frac{8}{\insertnumber{\prob}(\arraysize{\prob})^2} \sum_{\text{even }j\leq \numwin{}^k} |\inserts{j}^k|\frac{(\iskewfinal{j-1}^k)^2(\arraysize{\prob})^2}{(\window{}^k)^2 \beta^2} + \frac{2(\iskewfinal{\prob})^2}{(\arraysize{\prob})^2}
    \tag{by defn of $\arrayskew{j}^k$ and  $\sum |\inserts{j}^k| = \insertnumber{\prob}$}\\
    & \leq  \frac{8}{\insertnumber{\prob} \beta^2\window{}^k} \sum_{\text{even }j\leq \numwin{}^k} (\iskewfinal{j-1}^k)^2 + \frac{2(\iskewfinal{\prob})^2}{(\arraysize{\prob})^2} \tag{since $|\inserts{j}^k| \leq \window{}^k$} \\
    & \leq \frac{8\alpha 2^k}{\arraysize{\prob}\insertnumber{\prob}\beta^2} S^k + \frac{2(\iskewfinal{\prob})^2}{(\arraysize{\prob})^2} \tag{since $\window{}^k= \frac{\arraysize{\prob}}{\alpha 2^k}$} \\
    & \leq  \frac{8\alpha} {\arraysize{\prob}\insertnumber{\prob}\beta^2}\left(2^{k} S^k + \frac{1}{4}(\iskewfinal{\prob})^2\right). \tag{since $\alpha \ge \beta$ and thus $\arraysize{\prob}\geq \insertnumber{\prob}\alpha \geq \insertnumber{\prob} \beta^2 / \alpha$}
    \end{align*}

    Averaging over $k$ we get,
\begin{align*}
\E[(\Lambda(v))^2]& \leq  \frac{8\alpha} {\insertnumber{\prob}\arraysize{\prob}\beta^2}\left(\sum_{k=0}^{\imax}p_k2^{k}S^k
    +\sum_{k=0}^{\imax}p_k\frac{(\iskewfinal{\prob})^2}{4}\right) \\
    & \leq  \frac{8\alpha} {\insertnumber{\prob}\arraysize{\prob}\beta^2}\left(\sum_{k=1}^{\imax}S^k+\frac{S^0}{2}
    +\frac{(\iskewfinal{\prob})^2}{4}\right)\tag{since $p_k\leq 2^{-k}$, $p_0 \leq 1/2$ and $S^k \geq 0$} \\
    & \le  \frac{8\alpha}{\arraysize{\prob}\insertnumber{\prob}\beta^2}\sum_{k=0}^{\imax}S^k.\tag{since $S^0=(\iskewfinal{\prob})^2$}
\end{align*}
   
Finally, we analyze $\E[|\varepsilon(v)|]$. For each window $\inserts{j}^k$, and for each $v \in \inserts{j}^k$, we have by Claim \ref{clm:lambda} that
    
\begin{align*}
    |\varepsilon^k(v)| & \leq \frac{8|\arrayskew{j}^k|}{3\arraysize{\prob}}|\Lambda^k(v)|=\frac{8|\arrayskew{j}^k|\cdot |2(1-\free{\prob})\arrayskew{j}^k-\iskewfinal{\prob}|}{3(\arraysize{\prob})^2}\\
    &\leq 8\frac{2(\arrayskew{j}^k)^2 +|\arrayskew{j}^k||\iskewfinal{\prob}|}{3(\arraysize{\prob})^2} \tag{since $F_\prob \in [0, 1]$}\\
    &\leq 8\frac{\frac{5}{2}(\arrayskew{j}^k)^2 +\frac{1}{2}(\iskewfinal{\prob})^2}{3(\arraysize{\prob})^2} \hspace{.2 in} \tag{by the inequality $ab \leq (a^2+b^2)/2$}\\
    &\leq \frac{8(\arrayskew{j}^k)^2 +2(\iskewfinal{\prob})^2}{(\arraysize{\prob})^2},\\
\end{align*}
which is, quite fortuitously (and partly by design), the same as the upper bound on $(\Lambda^k(v))^2$ shown in (\ref{eqn:Delta^2}). Therefore, the exact same computation as for $\E[(\Lambda^k(v))^2]$ yields
the claimed bound.

\end{proof}

We now come to the critical part of the proof.
We have lower bounds on $\E[\Lambda(v)]$ in terms
of the sums $R^k$ and upper bounds
on $\E[(\Lambda(v))^2]$ and $\E[|\varepsilon(v)|]$
in terms of the sums $S^k$.  In order to complete
the proof we need to relate the quantities $R^k$
to the quantities $S^k$.  This connection
is provided by the following simple but crucial identity:

\begin{proposition}
\label{prop:tree}
For any $h < \ell$,
$$
S^h-S^{\ell} = 2\sum_{k=h+1}^{\ell} R^k.
$$
\end{proposition}

\begin{proof}
First we compute $S^{k}-S^{k+1}$. 
For each rebuild window $\inserts{j}^k = \inserts{2j-1}^{k+1} \cup \inserts{2j}^{k+1}$, since $\iskewfinal{j}^k = \iskewfinal{2j-1}^{k+1} + \iskewfinal{2j}^{k+1}$, we have
\[
(\iskewfinal{j}^k)^2=(\iskewfinal{2j-1}^{k+1})^2+(\iskewfinal{2j}^{k+1})^2+2\iskewfinal{2j-1}^{k+1}\iskewfinal{2j}^{k+1}.
\]
Summing both sides over $j$
yields $S^k= S^{k+1} + 2R^{k+1}$, so $S^{k}-S^{k+1}=2R^{k+1}$.  The desired equality follows by summing this equality for $k$ from $h$ to $\ell-1$.
\end{proof}

We note that this lemma is the reason that, in the specification of the algorithm, $\arrayskew{j}^k$ is defined to be 0 for odd
$j$. Had we applied the definition for even $j$ also to odd $j$, then in the lower bound of $\E[\Delta]$, instead of
$R^k = \sum_{\text{even } j} \iskewfinal{j-1}^{k}\iskewfinal{j}^{k}$ we would have to use ${R^k} = \sum_j \iskewfinal{j-1}^{k}\iskewfinal{j}^{k}$, where the sum is over all
$j$, not just even $j$. This change to the definition of $R^k$ would, in turn, cause the telescoping argument in  Proposition~\ref{prop:tree} to fail, and we would no longer be able
to relate the bound on $\E[\Lambda(v)]$ to the bounds on
$\E[(\Lambda(v))^2]$ and $\E[|\varepsilon(v)|]$.

We now manipulate the bound on $\E[\Delta(v)]$ to finish the proof of the lemma.
A key step is given by:
\begin{equation}
\label{eqn:key}
\sum_{k=1}^{\imax} kR^k = \sum_{k=1}^{\imax}\sum_{h=k}^{\imax}R^h = \sum_{k=0}^{\imax-1}\frac{S^k-S^{\imax}}{2} = \frac{1}{2}\sum_{k=0}^{\imax} S^k - \frac{\imax}{2} S^{\imax},
\end{equation}
where the second equality uses Proposition~\ref{prop:tree}.  As we will see below, this identity is the reason why $p_k$ was defined to be $2^{-(k + 1)} \cdot (1+k/\imax)$ rather than, say,
$2^{-k}$.  

We now lower bound $\E[\Lambda(v)]$ by
\begin{align}
    \E[\Lambda(v)] &= \frac{\alpha(1-\free{\prob})}{\insertnumber{\prob}\arraysize{\prob}\beta} \sum_{k=1}^{\imax} \left(1+\frac{k}{\imax}\right) R^k - \frac{S^0}
    {\insertnumber{\prob}\arraysize{\prob}} \tag{Prop.~\ref{prop:S^h}} \\
    & =\frac{\alpha(1-\free{\prob})}{\insertnumber{\prob}\arraysize{\prob}\beta} \left(\sum_{k=1}^{\imax}R^k +\frac{1}{\imax}\sum_{k=1}^{\imax}k R^k\right) - \frac{S^0}{\insertnumber{\prob}\arraysize{\prob}} \nonumber \\
    & =\frac{\alpha(1-\free{\prob})}{\insertnumber{\prob}\arraysize{\prob}\beta} \left(\frac{S^0-S^{\imax}}{2} +\frac{1}{2\imax}\sum_{k=0}^{\imax} S^k -  \frac{S^{\imax}}{2}\right) - \frac{S^0}{\insertnumber{\prob}\arraysize{\prob}}\tag{Prop.~\ref{prop:tree} and (\ref{eqn:key})} \\
    & \geq\frac{\alpha}{2\insertnumber{\prob}\arraysize{\prob}\beta} \left(\frac{S^0}{2} +\frac{1}{2\imax}\sum_{k=0}^{\imax} S^k\right) -  \frac{\alpha}{\insertnumber{\prob}\arraysize{\prob}\beta}S^{\imax} - 
    \frac{S^0}{\insertnumber{\prob}\arraysize{\prob}} \tag{$\free{\prob}\in [0,1/2]$, $S^k\geq 0$} \\
    & \geq \frac{\alpha}{2\insertnumber{\prob}\arraysize{\prob}\beta} \left(\frac{S^0}{2} +\frac{1}{2\imax}\sum_{k=0}^{\imax} S^k\right) -  \frac{\alpha}{\insertnumber{\prob}\arraysize{\prob}\beta}S^{\imax} - \frac{\alpha S^0}{8\beta \insertnumber{\prob}\arraysize{\prob}}\tag{since $\alpha\geq 8\beta$} \\
    & \geq \frac{\alpha}{2\insertnumber{\prob}\arraysize{\prob}\beta} \left(\frac{1}{2\imax}\sum_{k=0}^{\imax} S^k\right) -  \frac{\alpha}{\insertnumber{\prob}\arraysize{\prob}\beta}S^{\imax}  \tag{since $S^0 \ge 0$}\\
    & \geq \frac{\alpha}{2\insertnumber{\prob}\arraysize{\prob}\beta} \left(\frac{1}{2\imax}\sum_{k=0}^{\imax} S^k\right) -  \frac{1}{\beta}2^{-\imax}\tag{Prop.~\ref{prop:S^h}} \\
    & \geq \frac{\alpha}{4\insertnumber{\prob}\arraysize{\prob}\beta\imax} \sum_{k=0}^{\imax} S^k -  2^{-\imax}.  \label{eqn:Delta final}
 \end{align}

         We are now ready to complete the proof of Lemma \ref{lemma:main}. We have
\begin{align*}
    \E[\Delta(v)] & \geq  \E[\Lambda(v)]-\E[|\varepsilon(v)|] \tag{Prop.~\ref{prop:conversion}}\\
    & \geq  \frac{\alpha}{4\insertnumber{\prob}\arraysize{\prob}\beta\imax} \sum_{k=0}^{\imax} S^k -  2^{-\imax} - \frac{8\alpha } 
    {\arraysize{\prob}\insertnumber{\prob}\beta^2} \sum_{k=0}^{\imax}  S^k
    \hspace{.2 in} \tag{by \eqref{eqn:Delta final} and Lemma~\ref{lemma:expectation}}\\
    & = \left(\frac{\beta}{32\imax}- 1\right) \frac{8\alpha }{\arraysize{\prob}\insertnumber{\prob}\beta^2} \sum_{k=0}^{\imax}  S^k   -  2^{-\imax}  \\
    & \geq  \left(\frac{\beta}{32\imax}- 1\right)  \E[\Lambda(v)^2]   -  2^{-\imax}
    \tag{Lemma~\ref{lemma:expectation}}\\
    & \geq  \left(\frac{\beta}{32\imax}- 1\right)  \E[\Delta(v)^2]/2   -  2^{-\imax}
    \tag{Prop.~\ref{prop:conversion}}\\
    & \geq  \left(\frac{\beta}{100\imax}\right)  \E[\Delta(v)^2]   -  2^{-\imax},
    \tag{since $\beta \geq 1000 \imax$}
\end{align*}
which completes the proof of Lemma~\ref{lemma:main}.

\section{Proof of Claim \ref{claim:conc}}
\label{sec:claim}

In this section, we prove Claim \ref{claim:conc}, restated below:

\concclaim*

As notation, for a sequence $Z=Z_1,\ldots,Z_s$ of random variables, define
$\mu_1(Z),\dots,\mu_s(Z)$ and $V_1(Z),\ldots,V_s(Z)$ by:
\begin{eqnarray*}
\mu_j(Z) &=&\E[Z_j \mid Z_1,\ldots,Z_{j-1}]\\
V_j(Z)&=&\var[Z_j \mid Z_1,\ldots,Z_{j-1}].
\end{eqnarray*}
Also define $\mu(Z)=\sum \mu_j(Z)$, $V(Z)=\sum_j V_j(Z)$ and $\Sigma(Z)=\sum_j Z_j$. 
Note that $\mu_j(Z)$ and $V_j(Z)$ are random variables that are determined by the values of $Z_1,\ldots,Z_{j-1}$.  We will prove:
\begin{lemma} 
\label{lemma:conc}
Let $Z = Z_1, Z_2, \ldots, Z_{s}$ be random variables and suppose that $A$, $B$, and $C$ are positive real numbers
with $4A \leq B \leq 1$ such that,
with probability 1, we have for all $j  \in [1,s]$ that:
\begin{enumerate}
    \item $|Z_j| \le A$;
    \item $V_j(Z) \le  B \mu_j(Z) + C.$
\end{enumerate}
Then for $q \geq \max(2sC/B,16B)$, $\Pr[\Sigma(Z) < -q] \le 3e^{-q/16B} $.
\end{lemma}

Supposing $n$ is sufficiently large, Claim~\ref{claim:conc} follows from the lemma
with $q=0.2$, $s=r \le 2 \log n$, $A=\frac{3}{\beta}$, $B=100\imax/\beta$ and $C=2^{-\imax}=\frac{1}{\log^2n}$.  The hypothesis $4A \le B \le 1$ is satisfied since $\beta=\Omega((\log\log n)^2)$ and $\imax=2\log\log n$, and the hypothesis $q \geq \max(2sC/B,16B)$ is satisfied since $2sC/B=O((\log\log n) /\log n)$ and $16B=O(1/\log\log n)$
so both are smaller than $q=0.2$. The resulting
probability upper bound is $3e^{-.2\beta/1600\imax}=3e^{-\cbeta\log\log n /16000} \leq 3(\log n)^{-\cbeta/16000}$ and taking $\cbeta$ large enough this is less than $1/\log^3 n$, as required.

Recall that a martingale
is a sequence $X = X_1,X_2,\ldots,X_s$ of random variables such that, for any outcomes of $X_1, \ldots, X_{j - 1}$, we have $\mu_j(X) = \E[X_j \mid X_1, X_2, \ldots, X_{j - 1}] = 0$.
Lemma~\ref{lemma:conc} will be deduced from the following theorem of Freedman:

\begin{theorem}[Proposition 2.1 of \cite{Freed75}]
\label{thm:freedman}
Let $s\in \mathbb{N}$ and let $Y= Y_1,\ldots,Y_s$ be a martingale. Suppose  $D$ and
$v$ are positive real numbers such that, with probability 1:
\begin{enumerate}
    \item $|Y_j|\leq D$ for all $j$;
    \item $\sum_j V_j(Y)\leq v$.
\end{enumerate}
Then, for $\ell >0$,  $\Pr[\Sigma(Y) > \ell]\leq  e^{-\ell^2/2(v+D\ell)}$. 
\end{theorem}

A natural approach to proving Lemma~\ref{lemma:conc} is to define the martingale $X$ by $X_j=\mu_j(Z)-Z_j$. 
This would imply $V_j(X)=V_j(Z)$ for all $j$, and it would then suffice to upper bound
$\Pr[\Sigma(Z)<-q]=\Pr[\Sigma(X)  > \mu(Z)+q]$, which we might hope to do with Theorem~\ref{thm:freedman}.
However, the theorem cannot be applied directly because
the upper bound on $V_j(X)$ of $B\mu_j(Z)+sC$ implied by Hypothesis (2) of Lemma \ref{lemma:conc} is itself a random variable,
and the quantity
 $\mu(Z)+ q $ to which $\Sigma(X)$ is compared in the conclusion is also a random variable, while Theorem~\ref{thm:freedman} requires both to be fixed quantities.  To get around this we first prove:

\begin{proposition}
\label{prop:nice trick}
With $Z$  satisfying the hypotheses of Lemma~\ref{lemma:conc}, let $\rho_2>\rho_1$ be fixed real numbers.
For $q \in (-\rho_1,2\rho_2-\rho_1-\frac{2sC}{B}]$,
$$\Pr[\Sigma(Z)<-q \text{ and } \mu(Z) \in [\rho_1,\rho_2]] \leq e^{-(\rho_1+q)^2/(4B\rho_2)}.$$
\end{proposition}

\begin{proof}

Assume that $Z$ satisfies the hypotheses of Lemma~\ref{lemma:conc} and $q \in (-\rho_1,2\rho_2-\rho_1-\frac{2sC}{B}]$. We define sequences  $Z'$  and $X'$ which are modified versions of $Z$ and $X$. We'll be able to apply Theorem~\ref{thm:freedman} to $X'$.
Let $C_j(Z)$ denote the event:

$$\sum_{i = 1}^j V_i(Z) \le B\rho_2+sC,$$ 
and define
$Z'_1,\ldots,Z'_s$ to be given by:
$$Z'_j = \begin{cases}
    Z_j & \text{ if  $C_j(Z)$ holds,}\\
    0 & \text{ otherwise}.
\end{cases}$$

\begin{claim}
\label{claim:Z'}
$Z'$ has the following properties:
\begin{enumerate}
\item If $C_j(Z)$ holds
then $Z'_1,\ldots, Z'_j=Z_1,\ldots, Z_j$ and  $V_1(Z'),\ldots,V_j(Z')=V_1(Z),\ldots,V_j(Z)$.
\item $V(Z') \leq \rho_2B+sC$.
\item If $\mu(Z)\leq \rho_2$ then $C_j(Z)$ holds for all $j \in [1,s]$ and $Z'=Z$.
\end{enumerate}
\end{claim}
\begin{proof}
For the first claim, note that if $C_j(Z)$ holds, then 
conditioned on $Z_1,\ldots,Z_{j-1}$,
$Z'_j$ and $Z_j$ are equal as random variables and consequently
$V_j(Z')=V_j(Z)$.  Also, if $C_j(Z)$ holds then $C_i(Z)$ holds for all $i \leq j$,
since $V_i(Z) \geq 0$ for all $i$.   Therefore if $C_j(Z)$ holds $Z'_1,\ldots,Z'_j=Z_1,\ldots,Z_j$
and $V_1(Z'),\ldots,V_j(Z')=V_1(Z),\ldots,V_j(Z)$.

For the second claim, let $h$ be the least index for which $C_h(Z)$ fails, setting
$h=s+1$ if $C_j(Z)$ holds for all $j \leq s$.  Then $C_i(Z)$ fails for
$i \geq h$, and so conditioned on $Z_1,\ldots,Z_{h-1}$, for $i \geq h$,
$Z_i$ is identically 0, and so $V_i(Z)=0$. Therefore by part 1 of the claim:
\[
V(Z') \leq \sum_{i \leq s} V_i(Z') =\sum_{i\leq h-1} V_i(Z')=
\sum_{i \leq h-1} V_i(Z) \leq \rho_2B+sC,
\]
since $C_{h-1}(Z)$ holds.

For the third claim, if $\mu(Z) \leq \rho_2$  then by Hypothesis (2) of Lemma \ref{lemma:conc}, 
$\sum_{i\leq s}V_i(Z) \leq
\sum_{i=1}^s (B\mu_i(Z)+C) = B\mu(Z)+sC \leq B\rho_2+sC$ and so 
condition $C_s(Z)$ holds, which implies by the first part of the claim that $C_1(Z), \ldots, C_s(Z)$ all hold and that 
$Z'=Z$.
\end{proof}

Now define the martingale $X'$ by $X'_j=\mu_j(Z')-Z'_j$. Then:
\begin{align*}
\Pr[\Sigma(Z) < -q \text{ and } \mu(Z) \in  [\rho_1,\rho_2]] & \leq \Pr[\Sigma(Z') < -q \text{ and } \mu(Z') \in [\rho_1,\rho_2]] \tag{Part 3 of Claim~\ref{claim:Z'}}\\
& =  \Pr[\Sigma(X') > \mu(Z')+q \text{ and } \mu(Z') \in [\rho_1,\rho_2]] \\
&\leq \Pr[\Sigma(X') > \rho_1+q].
\end{align*}

We claim that the  hypotheses of Theorem~\ref{thm:freedman} are satisfied
with  $Y=X'$,  
$D=2A$, and $v=B\rho_2+sC$. For the first
hypothesis of Theorem \ref{thm:freedman}, $|Z_j|\leq A$ which 
implies $|\mu_j(Z)|\leq A$ and so $|X_j| \leq  |\mu_j(Z)|+|Z_j| \leq 2A$.
For the second hypothesis, we have 
\begin{align*}
V_j(X')&=\E[(\mu_j(X')-X'_j)^2 \mid X'_1,\ldots,X'_{j-1}] \\
&=\E[(\mu_j(Z')-Z'_j)^2 \mid X'_1,\ldots,X'_{j-1}] \\
&=\E[(\mu_j(Z')-Z'_j)^2 \mid Z'_1,\ldots,Z'_{j-1}]\\
&= V_j(Z') \leq B\rho_2+sC,  \tag{Part 2 of  Claim~\ref{claim:Z'}}
\end{align*}
where the 3rd equality holds because the sequences $Z'_1,\ldots,Z'_{j-1}$ and
$X'_1,\ldots,X'_{j-1}$ determine each other.

Since $\ell=\rho_1+q>0$ (by the hypothesis on $q$ in the current proposition),
we can apply Theorem~\ref{thm:freedman} to get:

\begin{align*}
\Pr\big[\Sigma(Z)<-q \text{ and } \mu \in [\rho_1,\rho_2]\big] &\leq \Pr\big[\Sigma(X') > \rho_1+q\big]\\ 
&\leq  \exp\left(-\frac{(\rho_1+q)^2}{2(B\rho_2+sC+2A(q+\rho_1))}\right) \tag{Theorem~\ref{thm:freedman}}\\
&\leq   \exp\left(-\frac{(\rho_1+q)^2}{2B(\rho_2+sC/B +(q+\rho_1)/2)}\right) \tag{since $4A \leq B$} \\
&\leq  \exp\left(-\frac{(\rho_1+q)^2}{4B\rho_2}\right), \tag{since $q\leq 2\rho_2-\rho_1-2sC/B$}
\end{align*}
as required.
\end{proof}

We can now finish the proof of Lemma~\ref{lemma:conc}.

\begin{proof}[Proof of Lemma~\ref{lemma:conc}]
Hypothesis 2 of the lemma implies $V(Z) \leq B\mu(Z)+sC$ and therefore
$\mu(Z) \geq -sC$,
since variance is nonnegative.
Cover the
interval $[-sC,\infty]$ by $[-sC,q] \cup \bigcup_{i \geq 1} [2^{i-1}q,2^iq].$  
For each of these intervals
we want to apply Proposition~\ref{prop:nice trick} with $[\rho_1,\rho_2]$ set
to that interval.
For $\rho_1=-sC$ and $\rho_2=q$, the hypothesis $q > -\rho_1$ holds because, by assumption, $q \ge 2sC / B > sC$; and the hypothesis $q \leq 2\rho_2-\rho_1-2sC/B = 2q+sC-2sC/B$ holds because $2sC/B\leq q$ by assumption.
For the interval $[\rho_1, \rho_2] = [2^{i - 1} q, 2^i q]$, the hypotheses of Proposition \ref{prop:nice trick} hold because $q > 0 > -\rho_1$ and  because $2\rho_2-\rho_1-2sC/B \geq 2\rho_2-\rho_1-q=(2^{i+1}-2^{i-1}-1)q \geq q$.  So, applying Proposition \ref{prop:nice trick} to each interval, we can conclude that:

\begin{align*}
\Pr[\Sigma(Z) < -q] &\leq  \Pr\big[\Sigma(Z) <-q  \text{ and } \mu(Z) \in [-sC,q]\big]+\sum_{i\geq 1} \Pr\left[\Sigma(Z) < -q \text{ and } \mu(Z) \in [q2^{i-1},q2^i]\right]\\
&\leq  \exp\left(-\frac{(q-sC)^2}{4Bq}\right)+\sum_{i\geq 1} \exp\left(-\frac{(q2^{i-1}+q)^2}{4\cdot2^iqB}\right) \tag{Prop.~\ref{prop:nice trick}}\\
&\leq  \exp\left(-\frac{(q-sC)^2}{4Bq}\right)+\sum_{i\geq 1} \exp\left(-\frac{(q2^{i-1})^2}{4\cdot2^iqB}\right) \\
 &\leq \exp\left(-\frac{(q/2)^2}{4Bq}\right)+\sum_{i\geq 1}\exp\left(-\frac{2^iq}{16B}\right) \tag{$q\geq 2sC/B$}\\
&\leq  \exp\left(-\frac{q}{16B}\right)+\sum_{j \ge 1}\exp\left(-\frac{jq}{16B}\right) \\
&\leq  \exp\left(-\frac{q}{16B}\right)+\sum_{j \ge 1}\exp\left(-\frac{q}{16B} - (j - 1)\right) \tag{$q>16B$}\\
&\leq \exp\left(-\frac{q}{16B}\right)+\exp\left(-\frac{q}{16B}\right)\sum_{j\geq 0}e^{-j} \\
&\leq  3 \exp\left(-\frac{q}{16B}\right) \tag{$q>16B$}.
\end{align*}
\end{proof}

\section{Related work}\label{sec:related}

In this section, we give a detailed discussion of related work on the list-labeling problem. To distinguish the different regimes in which one can study the problem, we will refer to $m = (1 + \Theta(1))n$ as the \defn{linear regime}, to $m = (1 + o(1))n$ as the \defn{dense regime}, to $m = n^{1 + \Theta(1)}$ as the \defn{polynomial regime}, and to $m = n^{\omega(1)}$ as the \defn{super-polynomial regime}. Although list labeling was originally formulated in the linear regime \cite{ItaiKoRo81}, the other regimes end up also being useful in many settings. 

\paragraph{Independent Formulations.}
There have been many independent formulations of list labeling under a variety of different names. 
The problem encapsulates several other scenarios beyond the maintenance of elements from an ordered universe in a sorted array.
Instead of elements coming from an ordered universe, one can think of elements coming from an unordered universe whose rank is determined relative to the elements in the current set at the moment of their insertion. 
This was the original formulation of Itai, Konheim, and Rodeh~\cite{ItaiKoRo81} who devised a sparse table scheme to implement priority queues.
Willard~\cite{Willard81} independently studied the \defn{file-maintenance problem} for maintaining order in a file as records are inserted and deleted.
Even more abstractly, one does not have to think of an array but of a linked list of items that are assigned labels from $\{1,\dots,m\}$,
and the natural order of the labels should correspond to the relative order of the items.
This view becomes relevant when $m$ is large relative to $n$ (the polynomial and super-polynomial regimes), and it was taken by Dietz~\cite{Dietz82}, Tsakalidis \cite{Tsakalidis84}, and Dietz and Sleator \cite{DietzSl87}, and Bender et al.~\cite{BenderCoDe02twosimplified} who (in some cases implicitly) applied both the polynomial and exponential regimes to the so-called \defn{order-maintenance problem}, which studies the abstract data-structural problem of maintaining ordered items in a linked list.
A problem similar to list labeling (in the polynomial regime) was studied in the context of balanced binary search trees by Andersson~\cite{Andersson89} and Andersson and Lai~\cite{AnderssonLa90}, as well as by Galperin and Rivest~\cite{GalperinR93} under the name \defn{scapegoat trees}. 
Raman~\cite{Raman99} formulated the problem in the linear regime in the context of building locality preserving dictionaries.
Hofri and Konheim~\cite{HofriK87} studied a sparse table structure that supports search, insert and deletion by keys in the linear and dense regimes.
Devanny, Fineman, Goodrich, and Kopelowitz~\cite{DevannyFiGo17} studied the \defn{online house numbering problem}, a version of list labeling where the goal is to minimize the maximum number of times that any one element gets moved (i.e., has its label changed).

\paragraph{Upper bounds.}
The most studied setting of the list-labeling problem is the linear regime, in which $m = (1 + \Theta(1)) n$. 
Itai, Konheim, and Rodeh\cite{ItaiKoRo81}, showed an upper bound of $O(\log^2 n)$ amortized cost per operation.
This was later deamortized to $O(\log^2 n)$ worst-case cost per operation by Willard~\cite{Willard82, Willard86, Willard92}.
Simplified algorithms for these upper bounds were provided by Katriel~\cite{Katriel02}, and Itai and Katriel~\cite{ItaiKa07} for the amortized bound and Bender, Cole, Demaine, Farach-Colton, and Zito~\cite{BenderCoDe02twosimplified} and Bender, Fineman, Gilbert, Kopelowitz, and Montes~\cite{BenderFiGi17} for the worst-case bound.
The upper bound of $O(\log^2 n)$ stood unimproved for four decades until Bender, Conway, Farach-Colton, Koml\'os, Kuszmaul, and Wein~\cite{BenderCFKKW22} showed an amortized $O(\log^{3/2}n)$ expected cost algorithm.
The same paper also proved an upper bound of $O(\log^{3/2}n / (\log^{1/2} \tau))$ for the \emph{sparse} regime where $m = \tau n$ for $\tau \leq n^{o(1)}$. The algorithm by Bender et al.~\cite{BenderCFKKW22} is history independent, and builds on techniques developed by an earlier $O(\log^2 n)$ expected-cost history-independent solution due to Bender, Berry, Johnson, Kroeger, McCauley, Phillips, Simon, Singh, and Zage~\cite{BenderBeJo16}.

In the polynomial regime, where $m=n^{1 + \Theta(1)}$, upper bounds of $O(\log n)$ have been shown~\cite{Kopelowitz12,Andersson89,GalperinR93}.
In the superpolynomial regime, where $m=n^{\omega(1)}$, Babka, Bul\'anek, \v{C}un\'at, Kouck\'y, and Saks~\cite{BabkaBCKS19} gave an algorithm with amortized $O(\log n/\loglog m)$ cost when $m=\Omega(2^{\log^k n})$, which implies a constant amortized cost algorithm in the pseudo-exponential regime where $m = 2^{n^{\Omega(1)}}$.

For the regime where $m=n$, Andersson and Lai~\cite{AnderssonLa90}, Zhang~\cite{zhang1993density}, and Bird and Sadnicki~\cite{BirdSa07} showed an $O(n\log^3n)$ upper bound for filling an array from empty to full (i.e., an insertion-only workload). This bound was subsequently improved to $O(n \log^{2.5} n)$ by Bender et al.~\cite{bender2022optimal}, and then to $\tilde{O}(\log^{2} n)$ in the current paper (Corollary \ref{cor:fillup}).

Finally, several papers (in the linear regime) have also studied forms of beyond-worst-case analysis. Bender and Hu~\cite{BenderHu07} provided an \emph{adaptive} solution, which has $O(\log n)$ amortized expected cost on certain common classes of instances while maintaining $O(\log^2n)$  amortized worst-case cost. McCauley, Moseley, Niaparast, and Singh \cite{mccauley2024online} study a setting in which one has access to a (possibly erroneous) prediction oracle, and give a solution that is parameterized by the oracle's error.

\paragraph{Lower bounds.}
In the linear regime, Dietz and Zhang~\cite{dietz1990lower} proved a lower bound of $\Omega(\log^2 n)$ for \defn{smooth} algorithms, which are restricted to rearrangements that spread a set of elements evenly across some subarray.
Bul\'anek, Kouck\'y, and Saks~\cite{BulanekKoSa12} later showed an $\Omega(\log^2 n)$ lower bound for deterministic algorithms.
Bender, Conway, Farach-Colton, Koml\'os, Kuszmaul, and Wein~\cite{BenderCFKKW22} showed a lower bound of $\Omega(\log^{3/2}n)$ for history-independent algorithms, where the notion of history independence that they used is that the set of slots occupied, at any given moment, should reveal nothing about the input sequence beyond the current number of elements.

In the polynomial regime, Dietz and Zhang~\cite{dietz1990lower} proved a lower bound of $\Omega(\log n)$ for smooth algorithms.
Dietz, Seiferas, and Zhang\cite{dietz2004tight}, and a later simplification by Babka, Bul\'anek, \v{C}un\'at, Kouck\'y, and Saks~\cite{BabkaBCKS12}, extended this to a lower bound of $\Omega(\log n)$ for general deterministic algorithms.
Finally, Bul\'anek, Kouck\'y, and Saks~\cite{BulanekKoSa13} proved an $\Omega(\log n)$ lower bound for general (including randomized) algorithms. This is also by extension the best known lower bound for randomized algorithms in the linear regime. 

In other regimes, Bul\'anek, Kouck\'y, and Saks~\cite{BulanekKoSa12} showed a deterministic lower bound of $\Omega(n\log^3n)$ for $n$ insertions into an initially empty array of size $m=n+n^{1-\epsilon}$. 
In the superpolynomial regime, Babka, Bul\'anek, \v{C}un\'at, Kouck\'y, and Saks~\cite{BabkaBCKS19} gave a deterministic lower bound of $\Omega\left(\frac{\log n}{\loglog m - \loglog n}\right)$ for $m$ between $n^{1+C}$ and $2^n$, which reduces to a bound of $\Omega(\log n)$ for $m=n^{1+C}$.

\paragraph{Other Theory Applications.}
In addition to the applications discussed above, list labeling has found many algorithmic applications in areas such as cache-oblivious data structures and computational geometry.
Many of these applications use \defn{packed-memory arrays}, which are list-labeling solutions in the linear (and dense) regimes with the added requirement that there are never more than $O(1)$ free slots in a row between consecutive elements. 
Various works show bounds of $O(\delta^{-1} \log^2 n)$ for this version of the problem~\cite{BenderDeFa05,BenderDeFa00,BenderFiGi05}. Improvements to list labeling in both \cite{BenderCFKKW22} and in this paper imply analogous improvements for packed-memory arrays (with our result bringing the bound down to $\tilde{O}(\delta^{-1}\log n)$). These improvements, in turn, imply immediate improvements to the bounds in many of the applications below. 

Packed-memory arrays have found extensive applications to the design of efficient cache-oblivious data structures.
Bender, Demaine, and Farach-Colton~\cite{BenderDeFa05} used the packed-memory array to construct a \defn{cache-oblivious B-tree}. Simplified algorithms for cache-oblivious B-trees were provided by Brodal, Fagerberg, and Jacob~\cite{BrodalFaJa02} and Bender, Duan, Iacono, and Wu~\cite{BenderDuIa04}. 
Bender, Fineman, Gilbert, and Kuszmaul~\cite{BenderFiGi05} presented \defn{concurrent} cache-oblivious B-trees and
Bender, Farach-Colton, and Kuszmaul~\cite{BenderFaKu06} presented \defn{cache-oblivious string B-trees}. 
All of these data structures use packed memory arrays. In each case, the list-labeling improvements in the current paper improve the range of parameters for which the above constructions are optimal, so that the restriction on the block-size $B$ goes from $B \ge \tilde{\Omega}(\log  \sqrt{\log n})$ (using the list-labeling solution from \cite{BenderCFKKW22}) to $B \ge \poly \log \log n$.

List labeling has also found applications in data structures for computational geometry problems.
Nekrich used the technique to design data structures for orthogonal range reporting~\cite{Nekrich07, Nekrich09} (these use the polynomial regime), the stabbing-max problem~\cite{Nekrich11} (this uses the linear regime), and a related problem of searching a dynamic catalog on a tree~\cite{Nekrich10} (this uses the linear regime). 
Similarly, Mortensen~\cite{Mortensen03} used the technique (in the linear regime) for the orthogonal range and dynamic line segment intersection reporting problems.

Additionally, Fagerberg, Hammer, and Meyer~\cite{FagerbergH019} use list labeling (implicitly, and in the linear regime) for a rebalancing scheme that maintains optimal height in a balanced B-tree. And Kopelowitz~\cite{Kopelowitz12} uses a generalization of the list-labeling problem (in the polynomial regime) to design an efficient algorithm for constructing suffix trees in an online fashion.

On the lower-bound side, Emek and Korman~\cite{EmekKo11} show how to make use of lower bounds for list labeling to derive lower bounds for the \emph{distributed controller problem}, which is a resource allocation problem in the distributed setting \cite{AfekAPS96}.

\paragraph{Practical Applications.}
Additionally, many practical applications make use of packed-memory arrays.  
Durand, Raffin and Faure~\cite{DurandRF12} use a packed-memory array in particle movement simulations to maintain sorted order for efficient searches. 
Khayyat, Lucia, Singh, Ouzzani, Papotti, Quian\'e-Ruiz, Tang and Kalnis~\cite{KLSOPQ0K17} handle dynamic database updates in inequality join algorithms using packed-memory arrays.
Toss, Pahins, Raffin and Comba~\cite{TossPRC18} constructed a \defn{packed-memory quadtree}, which supports large streaming spatiotemporal datasets. 
De Leo and Boncz~\cite{LeoB19pma} implement a \defn{rewired memory array}, which improves the practical performance of packed-memory arrays. 
\defn{Parallel} packed-memory arrays have been implemented in several works \cite{WheatmanX21, WheatmanX18, WheatmanB21, PandeyWXB21, LeoB21, LeoB19fastconcurrent} to store dynamic graphs with fast updates and range queries.

\section{Acknowledgements}

This work was supported by NSF grants 
CCF-2106999, 
CCF-2118620, 
CNS-1938180, 
CCF-2118832,  
CCF-2106827, 
CNS-1938709, and
CCF-2247577. 

\noindent Hanna Koml\'os was partially supported by the Graduate Fellowships for STEM Diversity.

\noindent Michal Kouck{\'{y}} carried out part of the work during an extended visit to DIMACS, with support from the National Science Foundation under grant number CCF-1836666 and from The Thomas C. and Marie M. Murray Distinguished Visiting Professorship in the Field of Computer Science at Rutgers University.
He was also partially supported by the Grant Agency of the Czech Republic under the grant agreement no. 24-10306S. 
This project has received funding from the European Union’s Horizon 2020 research and innovation programme under the Marie Skłodowska-Curie grant agreement No. 823748 (H2020-MSCA-RISE project CoSP). 

\noindent William Kuszmaul was partially supported by the Harvard Rabin Postdoctoral Fellowship.

\bibliographystyle{plain}
\bibliography{references}

\appendix
\section{Reducing Theorem \ref{thm:actualmain} to \ref{thm:nicetoprove}}\label{app:reductions}

To reduce Theorem \ref{thm:actualmain} to Theorem \ref{thm:nicetoprove}, we will make a series of (standard) simplifications that are each without loss of generality.

\paragraph{Ignoring deletions. }
We may assume without loss of generality that the sequence of operations includes only insertions.

\begin{proposition}\label{prop:deletions}
Any list-labeling solution that can start with (up to) $(1 - 3\gamma) m$ elements and support $\gamma m$ insertions with amortized expected cost $O(t(m, \gamma))$, can be modified to handle an arbitrary sequence of insertions/deletions, with up to $n = (1 - \delta)n$ elements present at a time, and with amortized expected cost $O(t(m, \delta/3) + 1 / \delta)$ per operation.
\end{proposition}
\begin{proof}
Set $\gamma = \delta / 3$. We can collect deletions into batches of size $\gamma m$. As a batch forms, we ``pretend'' that the elements in the batch have not yet been deleted (i.e., we replace the deleted elements with tombstones, which we think of as elements).

Once a batch is fully formed, we rebuild the entire data structure from scratch, so that the deleted elements are cleared out. This rebuild increases the amortized expected cost by $O(1 / \delta)$ per operation. During each batch, we are supporting an insertion-only workload that starts with (up to) $(1 - \delta)m = (1 - 3\gamma) m$ elements and performs up to $\delta n / 3 = \gamma m$ insertions. The amortized expected cost per batch is therefore $O(t(m, \gamma))$.
\end{proof}

\paragraph{Reducing to $n = m/2$. }Our new task is to support a sequence of insertions that starts with up to $(1 - 3\delta)n$ elements, and performs up to $\delta n$ insertions. 

The following lemma reduces this problem to the problem of performing $n = m/2$ insertions in an initially empty size-$m$ array. 

\begin{lemma}\label{lemma:increment}
    Let there be a list labeling algorithm $A'$ that for every $n' \ge 1$, it can insert $n'$ items into an initially empty array of size $m'=2n'$ for amortized expected cost $t(n')$, where $t(n')\ge 1$ is a non-decreasing function.
    Then for every fixed $\delta \in (0,1/2]$ there is a list labeling algorithm $A$ that for every $m \ge 1$ can insert $\lceil \delta m /3 \rceil$ items into an array of size $m$ that already contains $\lfloor (1-\delta)m \rfloor$ items, and where the amortized expected cost is at most cost $O(t(m) / \delta + 1 / \delta)$.
\end{lemma}

As the proof of Lemma \ref{lemma:increment} requires some care, we defer it to the end of the section.

\paragraph{Starting with $m / 4$ elements. }So far, we have reduced Theorem \ref{thm:actualmain} to the setting in which we wish to perform $n = m/2$ insertions in an initially-empty array of size $m$.  However, we can break these insertions into batches, where we fill the array from from $1 / 2^i$ full to $1 / 2^{i - 1}$ full, for some $i$, and we can implement each batch on an array of size $m / 2^{i - 2} \le m$. Thus, if we focus just on the task of implementing a batch, our final problem is: perform $m/4$ insertions in an array that initially contains $m/4$ elements. This is precisely the problem considered by Theorem \ref{thm:nicetoprove}, which completes the reduction from Theorem \ref{thm:actualmain}.

\paragraph{Proving Lemma \ref{lemma:increment}. }
We now prove Lemma \ref{lemma:increment}. 
\begin{proof}
For $\delta < 12/m$ the claim is trivial so we assume that $\delta \ge 12/m$.
Our algorithm $A$ with a \defn{real} array of size $m$ will simulate algorithm $A'$ on a \defn{virtual} array of size $m'=2n'$, where $n' = 2 \lfloor \delta m /3 \rfloor$.
Algorithm $A'$ will get to insert $n'$ items into its virtual array.
The first $n'/2$ items that will $A'$ get are selected from the initial items that are in the real array of $A$.
The next $n'/2$ items will be the items that $A$ should insert into its real array (except for the very last one depending on rounding).
The state of the real array during the latter $n'/2$ insertions will reflect the state of the virtual array.

$A$ will classify each of its items as either \defn{visible} or \defn{invisible}.
All items that will be inserted into the virtual array will be visible, all the other items will be invisible.
In particular, all the items newly inserted into the real array will be visible.
Initially, $A$ selects from the real array $n'/2$ items as the visible items and declares the remaining items as invisible.
The algorithm selects as visible each initial item of rank $1+ \lceil 3/\delta \rceil i$, for $i=0,1,\dots$,
together with additional items of the smallest rank so to have exactly $n'/2$ visible items.
(As the number of initial items is at least $\lfloor m/2 \rfloor \ge n'/2$, there are enough items to chose from.)

Algorithm $A$ will maintain the following two invariants: 
(1) No free slot in the real array can be immediately to the left of an invisible item, and
 (2) If we remove the invisible items together with their slots from the real array we get a copy of the current state of the virtual array. 
Since the left-most item in $A$ will be always visible, invariant (1) means that invisible items form blocks of invisible items that follow immediately a visible item.
After each block of invisible items there might be free slots followed by a visible item.
Because we initially select each item of rank  $1+ \lceil 3/\delta \rceil i$, for $i=0,1,\dots$, as visible
and newly inserted items will be also visible, each block of invisible items will always be of size at most  $\lceil 3/\delta \rceil - 1$.

To start the simulation, $A$ inserts the initial set of visible items into the virtual array using $A'$.
Then it will rearrange the real array to satisfy the two invariants.
This will move at most $m$ items in the real array.
Then we process new insertions into $A$.

For each newly inserted item $b$, $A$ proceeds as follows.
It passes $b$ to $A'$ as a new insertion.
In response to the insertion request, $A'$ might rearrange its items in the virtual array to prepare an appropriate free slot for $b$.
Then $A'$ inserts $b$ into the free slot.
Before $A'$ inserts $b$ into the free slot, $A$ rearranges its real array to satisfy invariant (2) (and also invariant (1)) as items in the virtual array might have moved.
Notice, the position of a particular visible item in the real array is given by the number of visible items to its left, together with the number of empty slots to its left,
and the number of invisible items to its left. 
Similarly, the position of the same visible item in the virtual array is given by the number of visible items to its left together with the number of empty slots to its left.
This implies that if an item in the virtual array retains its position during the rearrangement by $A'$, 
it should retain its position also in the real array during the rearrangement by $A$.
Also the block of invisible items following such an item will stay in place.

Thus $A$ will have to move at most  $\lceil 3/\delta \rceil$-times many items as $A'$ did in the virtual array in order to re-establish the invariants.
(It has to move the same number of visible items and each is followed by a block of at most  $\lceil 3/\delta \rceil - 1$ invisible items.)

After the rearrangement of the real array, $A$ will proceed to insert the item $b$.
Let $a$ be the closest visible item before $b$ in the virtual array.
Let $b$ be put into $i$-th empty slot following $a$ in the virtual array.
Let there be $\ell$ invisible items following $a$ in the real array.
Let $\ell'$ of those invisible items be smaller than $b$.
Algorithm $A$ will move the last $\ell-\ell'$ invisible items following immediately after $a$ in the real array $i$ positions to the right.
Then $A$ inserts $b$ into $(\ell'+i)$-th position after $a$, that is in the free slot immediately to the left of the moved invisible items.
This will re-establish the correspondence between the virtual and real array.
The cost of the additional moves is at most  $\lceil 3/\delta \rceil$.

The total number of moves done by $A'$ during its $n'$ insertions is $n' \cdot t(n')$.
(Although only half of the inserted items are new.)
Hence, the total number of moves done by $A$ during $n'/2$ new insertions is 
bounded by $m +  \lceil 3/\delta \rceil \cdot n' \cdot t(n') +  \lceil 3/\delta \rceil \cdot n'/2$.
Since  $\lceil 3/\delta \rceil \cdot \lfloor \delta m /3 \rfloor \le \frac{\delta m}{3} \cdot \frac{3+\delta}{\delta} \le 2m$,
the total cost can be bounded by $4m t(m) + 3m$.

We can accommodate an additional insert into $A$ for the cost of at most $m$,
hence inserting at least  $\delta m /3$ items
for the amortized expected cost $3(4m t(m) + 4m) / \delta m = 12(t(m)+1) / \delta$ as needed.
\end{proof}

Finally, it is worth pointing out one corollary of the lemma, which is the following claim about filling an array from empty to full:

\begin{corollary}
    If there is a list labeling algorithm $A'$ that for every $n' \ge 1$ can insert $n'$ items into an initially empty array of size $m'=2n'$ for amortized expected cost $t(n')$, where $t(n')\ge 1$ is a non-decreasing function
    then there is a list labeling algorithm $A$ that can insert $n$ items into an initially empty array of size $n$ with amortized expected cost $O(t(n) \log n)$ per insertion.
\end{corollary}

\begin{proof}
First, apply the algorithm $A'$ to insert $\lfloor n/2 \rfloor$ items into the array for the total cost at most $n \cdot t(n)$.
Then proceed in phases. 
Each phase $i=1,\dots$, starts with $e_i\ge 1$ remaining empty slots.
It applies algorithm $A$ from Lemma~\ref{lemma:increment} for $\epsilon_i = e_i/n $ to insert next $\lceil e_i /3 \rceil$ items.
The algorithm stops once $n$ items are inserted.
The cost of each phase $i\ge 1$ is at most $\lceil \frac{e_i}{3} \rceil \cdot  12(t(n)+1) / \delta = \lceil \frac{e_i}{3} \rceil \cdot \frac{n}{e_i} \cdot 12(t(n)+1) \le 12n (t(n)+1)$.
Since $e_{i+1} \le 2 e_i/3$, there are at most $\log_{3/2} n$ phases.
Thus the total cost to fill in the array is bounded by $O(n \cdot t(n) \log n)$.
The lemma follows.
\end{proof}

Thus, one immediate consequence of Theorem \ref{thm:nicetoprove} is:

\fillup*

\section{Pseudocode for the See-Saw Algorithm} 
\label{app:pseudocode}

In this section, we give pseudocode for the See-Saw Algorithm. We assume parameters $\alpha = \calpha (\loglog n)^2$ and $\beta = \cbeta (\loglog n)^2$, where $\calpha$ and $\cbeta$ are positive constants selected so that $\calpha$, $\cbeta$, and $\calpha / \cbeta$ are all sufficiently large.

\paragraph{Variables to be used in pseudocode. }Before presenting the algorithm pseudocode, we list the relevant variables for \defn{subproblem} $\prob$. 
We emphasize that many of these variables are dynamically changing over time, i.e., are updated dynamically within the pseudocode. 

\begin{itemize}

    \item $\leftchild{\prob}$ and $\rightchild{\prob}$ are the \defn{left and right child} of $\prob$, respectively.
    
    \item $\arr{\prob}$ is an array such that $\arr{\prob} = \arr{\leftchild{\prob}} \oplus \arr{\rightchild{\prob}}$ (the concatenation of the arrays). 

    \item $\arrayskew{\prob}$ is the \defn{array skew}, such that 
$|\arr{\leftchild{\prob}}| = |\arr{\prob}|/2-\arrayskew{\prob}$ and $|\arr{\rightchild{\prob}}|= |\arr{\prob}|/2+\arrayskew{\prob}$.

    \item The \defn{pivot} $\pivot{\prob}$ partitions the insertions that go to the left and right children of $\prob$.
    Upon creation of a subproblem, $\pivot{\prob}$ will be set to be the largest element stored in the subarray of its left child.
    \footnote{It turns out that $\leftchild{\prob}$ is guaranteed to have at least one element, so $\pivot{\prob}$ is guaranteed to exist. Since we are not going to prove this explicitly, one can think of there as being an extra edge case (that will never occur) in which, if $\leftchild{\prob}$ has no elements, then insertions to $\prob$ always go to $\rightchild{\prob}$.}
    This element will remain the pivot until $\prob$ ends or is reset.
 
    \item The \defn{rebuild window size} $\window{\prob}$ is the number of insertions permitted between rebuilds.

    \item $\insertcount{\prob}$ is the number of insertions that have occurred during the current rebuild window.

    \item $\lifetimeinsertcount{\prob}$ 
    is the number of insertions that have occurred during the lifetime of $\prob$.

    \item $\insertskew{\prob}$ is the number of insertions that occurred during the current window that are greater than the pivot minus those that are less than the pivot.  This is called the \defn{insertion skew} of the window.

    \item $\windowcount{\prob}$ is a counter specifying which window we are in, starting with window 1. 

\end{itemize}
\smallskip

\paragraph{Pseudocode.} Below, we give pseudocode for both insertions and the subroutines used within an insertion. We assume that the subproblem tree is initialized (at the beginning of time, with $m/4$ initial elements) by a call to \textsc{CreateSubtree}.

\bigskip
\algo{CreateSubtree}($\arr{}',\mathcal{S}'$)
\begin{algorithmic}[1]
\State Move the items in $\mathcal{S}'$ so that they are uniformly spread out in array $\arr{}'$
\State \Return \textsc{AllocateBalancedSubproblems}($\arr{}'$) \Comment{Builds tree of subproblems on $\arr{}'$}
\end{algorithmic}
\label{alg:rebuild}

\bigskip
\algo{AllocateBalancedSubproblems}($\arr{}'$):
\begin{algorithmic}[1]
    \State Create a new subproblem $\prob$   
    \State $\arr{\prob} \gets \arr{}'$
    \If{density$(\prob) > 0.75$ or $|\arr{\prob}| \le 2^{\sqrt{\log n}}$} 
    \State Declare $\prob$ to be a leaf
    \State \Return $\prob$
    \EndIf
 
    \State $\window{\prob} \gets$ \textsc{PickWindowLength}($\prob$)
    \State $\insertcount{\prob} \gets 0$ ; $\insertskew{\prob} \gets 0$ ; $\windowcount{\prob} \gets 0$; $\lifetimeinsertcount{\prob} \gets 0$ 
    \State $L \gets$ the left half of $\arr{\prob}$, $R \gets$ the right half of $\arr{\prob}$
    \State $\pivot{\prob} \gets$ the largest element in $L$  
    \State $\leftchild{\prob} \gets$ \textsc{AllocateBalancedSubproblems}($\arr{\leftchild{\prob}}$) 
    \State $\rightchild{\prob} \gets$ \textsc{AllocateBalancedSubproblems}($\arr{\rightchild{\prob}}$)
    \State \Return $\prob$
\end{algorithmic}
\label{alg:initialization}
\algo{Insert}$(x,\prob)$:
\begin{algorithmic}[1]
\If{$\prob$ is a leaf}
    \State Insert $x$ into $\prob$ using the classical algorithm
    \State \Return
\EndIf
\If{$x \leq \pivot{\prob}$}
    \State \textsc{Insert}$(x,\leftchild{\prob})$
    \State $\insertskew{\prob} \gets \insertskew{\prob} - 1$
    \If{$\lifetimeinsertcount{\leftchild{\prob}} \ge |\arr{\leftchild{\prob}}|/\alpha$}
    \State $\leftchild{\prob} \gets $ \textsc{CreateSubtree}($\arr{\leftchild{\prob}}$, \textsc{Set}(${\leftchild{\prob}}$)
    \Comment{Reset $\leftchild{\prob}$}
    \EndIf
\Else   
    \State \textsc{Insert}$(x,\rightchild{\prob})$
    \State $\insertskew{\prob} \gets \insertskew{\prob} + 1$
    \If{$\lifetimeinsertcount{\rightchild{\prob}} \ge |\arr{\rightchild{\prob}}|/\alpha$}
    \State $\rightchild{\prob} \gets $ \textsc{CreateSubtree}($\arr{\rightchild{\prob}}$, \textsc{Set}(${\rightchild{\prob}}$)
    \Comment{Reset $\rightchild{\prob}$}
    
    \EndIf
\EndIf
\State $\insertcount{\prob} \gets \insertcount{\prob} + 1$
\State $\lifetimeinsertcount{\prob} \gets \lifetimeinsertcount{\prob} + 1$
\If{$\insertcount{\prob} = \window{\prob}$} \Comment{End of rebuild window}
    \State \textsc{SkewRebuild}($\prob$) 
    \State $\windowcount{\prob} \gets \windowcount{\prob} + 1$;
     $\insertcount{\prob} \gets 0$;
     $\insertskew{\prob} \gets 0$
\EndIf

\If{$\prob = \mathrm{root}$ and $\insertcount{\prob} = m/\alpha$}
    \State root $\gets$ \textsc{CreateSubtree}($\arr{},\textsc{Set}(\mathrm{root})$) \Comment{Reset the root}
\EndIf
\end{algorithmic}
\label{alg:insert}

\bigskip

\algo{SkewRebuild}$(\prob)$:
\begin{algorithmic}[1]
\State $\arrayskew{\prob} \gets $ \textsc{PickArraySkew}($\prob$)
\State $\tempset{L} \gets $\textsc{Set}(${\leftchild{\prob}}$) 
\State $\tempset{R} \gets $\textsc{Set}(${\rightchild{\prob}}$)
\Comment{Keep the items stored in the left and right children the same}
\State $L \gets$ the array consisting of the first $|\arr{\prob}| - \arrayskew{\prob}$ slots in $\arr{\prob}$ 
\State $R \gets$ the array consisting of the first $|\arr{\prob}| + \arrayskew{\prob}$ slots in $\arr{\prob}$
\State $\leftchild{\prob} \gets $ \textsc{CreateSubtree}($L$, $\tempset{L}$)
\State $\rightchild{\prob} \gets $\textsc{CreateSubtree}($R$, $\tempset{R}$)
\end{algorithmic}
\label{alg:skewrebuild}

\bigskip

\algo{PickArraySkew}($\prob$):
\begin{algorithmic}[1]
\If{$\windowcount{\prob}$ is odd}
   \State \Return 0
\Else   
   \State \Return $|\arr{\prob}| \cdot \frac{\insertskew{\prob}}{\beta \window{\prob}}$ 
\EndIf
\end{algorithmic}
\label{alg:pickarrayskew}
\algo{PickWindowLength}($\prob$):
\begin{algorithmic}[1]
\State $\imax \gets 2 \log \log n$
\State For $k \in [1, \imax]$, $p_k \gets 2^{-(k + 1)} (1 + k/\imax)$
\State $p_0 \gets 1 - \sum_{k = 1}^{\imax} p_i$
\State Draw $\windowpar{\prob}$ so that $\Pr[\windowpar{\prob} = k] = p_k$
\State \Return $|\arr{\prob}|/(\alpha 2^{\windowpar{\prob}})$
\end{algorithmic}
\label{alg:pickwindowlength}

\bigskip
\algo{Set}($\prob$):
\begin{algorithmic}[1]
\State \Return $\{y \mid y \textrm { is stored in } \arr{\prob} \}$
\end{algorithmic}
\label{alg:set}

\end{document}